\documentclass{article}
\usepackage{amsmath}
\usepackage{fancyhdr}
\usepackage{amssymb}
\usepackage{amsthm}
\usepackage{graphicx}
\usepackage{varioref}
\usepackage{verbatim} 
\usepackage{multicol}
\usepackage{enumerate}
\usepackage[normalem]{ulem}
\usepackage{caption}
\usepackage{subcaption}
\usepackage[T1]{fontenc}
\usepackage[margin=1in]{geometry}
\usepackage{thm-restate}

\usepackage{mathrsfs}

\usepackage{url}
\usepackage{xspace}
\usepackage{thm-restate}

\usepackage{mathpazo}

\usepackage{microtype}

\usepackage{tikz}
\usetikzlibrary{calc}
\usetikzlibrary{patterns}
\usetikzlibrary{arrows}

\usepackage{epstopdf}

\usepackage[colorlinks = true]{hyperref}
\usepackage{xcolor}
\definecolor{darkred}  {rgb}{0.5,0,0}
\definecolor{darkblue} {rgb}{0,0,0.5}
\definecolor{darkgreen}{rgb}{0,0.5,0}
\hypersetup{
  urlcolor   = blue,         
  linkcolor  = blue,     
  citecolor  = darkgreen,    
  filecolor  = darkred       
}

\usepackage{cleveref}
\crefname{lemma}{Lemma}{Lemmas}
\crefname{proposition}{Proposition}{Propositions}
\crefname{definition}{Definition}{Definitions}
\crefname{theorem}{Theorem}{Theorems}
\crefname{conjecture}{Conjecture}{Conjectures}
\crefname{corollary}{Corollary}{Corollaries}
\crefname{section}{Section}{Sections}
\crefname{appendix}{Appendix}{Appendices}
\crefname{figure}{Figure}{Figures}
\crefname{equation}{Eq.}{Eqs.}
\crefname{table}{Table}{Tables}
\crefname{claim}{Claim}{Claims}
\crefname{algorithm}{Algorithm}{Algorithms}
\usepackage{qcircuit}

\newtheorem{theorem}{Theorem}
\newtheorem{lemma}[theorem]{Lemma}

\newtheorem{definition}[theorem]{Definition}
\newtheorem{corollary}[theorem]{Corollary}

\newtheorem*{conjecture*}{Conjecture}

\newtheorem{claim}[theorem]{Claim}
\theoremstyle{definition}

\newtheorem{algorithm}[theorem]{Algorithm}

\usepackage{authblk}

\newcommand{\ket}[1]{|#1\rangle}
\newcommand{\bra}[1]{\langle#1|}

\newcommand{\tth}[0]{\textsuperscript{th}}

\DeclareMathOperator*{\argmin}{arg\,min}

\DeclareMathAlphabet{\matheu}{U}{eus}{m}{n}

\newcommand{\sop}[1]{{\mathcal #1}}

\newcommand{\braket}[2]{\langle{#1}|{#2}\rangle}

\newcommand{\eps}{\epsilon}

\newcommand{\edgeL}{\ell}
\newcommand{\Ohm}{c}
\newcommand{\norm}[1]{\left\| #1 \right\|}
\newcommand{\abs}[1]{\left| #1 \right|}

\def\tO{\widetilde{\mathrm{O}}}

\begin{document}

\title{Quantum Algorithms for Connectivity and Related Problems}
\author[1]{Michael Jarret}
\affil[1]{Perimeter Institute}
\author[2]{Stacey Jeffery}
\affil[2]{Qusoft, CWI}
\author[3]{Shelby Kimmel}
\affil[3]{Middlebury College}
\author[2]{Alvaro Piedrafita}

\date{}

\maketitle

\begin{abstract}
An important family of span programs, $st$-connectivity span programs, have
been used to design quantum algorithms in various contexts, including a number
of graph problems and formula evaluation problems. The complexity of the
resulting algorithms depends on the largest \emph{positive witness size} of
any 1-input, and the largest \emph{negative witness size} of any 0-input.
Belovs and Reichardt first showed that the positive witness size is exactly
characterized by the effective resistance of the input graph, but only rough
upper bounds were known previously on the negative witness size.
We show that the negative witness size in an $st$-connectivity span program is
exactly characterized by the \emph{capacitance} of the input graph. This gives
a tight analysis for algorithms based on $st$-connectivity span programs on
any set of inputs.

We use this analysis to give a new quantum algorithm for estimating the
capacitance of a graph. We also describe a new quantum algorithm for deciding if a graph is
connected, which improves the previous best quantum algorithm for this problem
if we're promised that either the graph has at least $\kappa>1$ components, or the graph is connected and has small \emph{average resistance}, which is upper bounded by the diameter. We also give an alternative algorithm for deciding if a graph is
connected that can be better than our first algorithm when the maximum degree is small. Finally, using ideas from our second connectivity algorithm, we give an algorithm for estimating the \emph{algebraic connectivity} of a graph, the second largest eigenvalue of the Laplacian. 
\end{abstract}


\section{Introduction}

Span programs are an algebraic model of computation first developed by
Karchmer and Wigderson \cite{KW93} to study classical logspace complexity, and
introduced to the study of quantum algorithms by Reichardt and Sp\v{a}lek
\cite{RS12}.  In \cite{R01,Rei09}, Reichardt used the concept of span programs
to prove that the general adversary bound gives a tight lower bound on the
quantum query complexity of any given decision problem, thus showing the deep
connection between span programs and quantum query algorithms.

Given a span program, a generic transformation compiles it into a quantum
algorithm, whose query complexity is analyzed by taking the geometric mean of
two quantities: the largest \emph{positive witness size} of any 1-input; and
the largest \emph{negative witness size} of any 0-input. Thus, in order to
analyze the query complexity of an algorithm obtained in this way, it is
necessary to characterize, or at least upper bound, these quantities. Moreover, there is always a span program based algorithm with asymptotically optimal quantum query complexity \cite{R01,Rei09}.

The relationship between quantum query algorithms and span programs is potentially a powerful tool, but this correspondence alone is not a
recipe for finding such an algorithm, and actually producing an optimal (or even good) span
program for a given problem is generally difficult. Despite this difficulty, a
number of span programs have been found for important problems such as $k$-distinctness
\cite{Bel12}, formula evaluation \cite{RS12,reichardt2010span}, and
$st$-connectivity \cite{BR12}. The latter span program is of particular
importance, as it has been applied to a number of graph problems
\cite{cade2016time}, to generic formula evaluation problems \cite{JK2017}, and
underlies the learning graph framework \cite{Bel11}. The $st$-connectivity based
algorithms are also of interest because, unlike with generic span program
algorithms, it is often possible to analyze not only query complexity, but
also the time complexity.

While span program algorithms are universal for quantum query algorithms, it can also be fruitful to analyze the unitaries used in these algorithms in ways that are different from how they appear in the standard span program algorithm. For example, Ref.~\cite{IJ15} presents an algorithm to estimate span program witness sizes based on techniques from the standard span program algorithm. We will take a similar approach in this paper, deriving new algorithms based on unitaries from the span program algorithm for $st$-connectivity.

The problems of $st$-connectivity and connectivity will be considered in this
paper. For a family of undirected graphs $G$ on $N$ edges, for
$N\in\mathbb{N}$, and vertex set containing $s$ and $t$, the problem
$st$-$\textsc{conn}_G$ is the following: Given $x\in\{0,1\}^{E(G)}$, decide if
there is a path from $s$ to $t$ in $G(x)$, where $G(x)$ is the subgraph of $G$
obtained by including an edge $e$ if $x_e=1$\footnote{We consider more
complicated ways of associating edges with input variables in
\cref{sec:graphPrelim}, but the basic idea is captured by this simpler
picture.}. Similarly, the problem of $\textsc{conn}_G$ is the following: Given
$x\in\{0,1\}^{E(G)}$, determine if every vertex in $G(x)$ is connected to
every other vertex in $G(x)$.

\subsection{Contributions}

In all of the following problems, we assume we have access to a black box
unitary $O_x$ that tells us about the presence or absence of edges in a graph
$G$, and the query complexity refers to the number of uses of $O_x$ to solve a problem with high probability. 

\paragraph{Characterizing the negative witness of $st$-connectivity span programs} 
An important span program for solving solving $st$-connectivity in subgraphs
of complete graphs without edge weights was presented in Ref.~\cite{BR12}.
When generalized to subgraphs of arbitrary weighted graphs\footnote{Assigning
positive weights to edges does not change whether or not a graph is
$st$-connected, but rather, the weights should be considered as parameters of
the span program that affect its complexities.}, this span program was applied to develop the learning graph framework \cite{Bel11},
and quantum algorithms for formula evaluation~\cite{JK2017}.

In Ref.~\cite{BR12}, Belovs and Reichardt gave a tight characterization of the
positive witness size as the \emph{effective resistance} between $s$ and $t$
in the input graph. However, for the negative witness size of an input in
which $s$ and $t$ are not connected, they gave only a rough upper bound of
$n^2$, which is refined in \cite{Bel11} to the total weight of an $st$-cut,
which is still not a tight bound. In Ref.~\cite{JK2017}, it was shown that
when the parent graph is planar and $s$ and $t$ are on the same face, the
negative witness can be characterized exactly as the effective resistance of a
graph related to the planar dual of the parent graph. 
In particular, this allowed for a tight analysis of $st$-connectivity-based span program algorithms for
formula evaluation in \cite{JK2017}.

In this work, we bring the story to its conclusion, by showing that the
negative witness size of the $st$-connectivity span program is exactly
characterized by the \emph{effective capacitance} of the input graph (\cref{thm:negwit-capacitance}). At a
high-level, this well-studied electrical network quantity is a measure of the
potential difference that the network could store between the component
containing $s$ and the component containing $t$. The more, shorter paths
between these two components in the graph $G\setminus G(x)$, the greater the
capacitance. This characterization tells us
that quantum algorithms can quickly decide $st$-connectivity on graphs that
are promised to have either small effective resistance or small effective
capacitance.

\paragraph{Quantum algorithm for estimating $st$-capacitance} As one immediate
application, we get a new quantum algorithm for estimating the capacitance of
an input graph $G(x)$ to multiplicative error $\varepsilon$, with query complexity
$\widetilde{O}(\varepsilon^{-3/2}\sqrt{C_{s,t}(G(x))p})$, where $C_{s,t}(G(x))$
is the $st$-capacitance of $G(x)$, and $p$ is the length of the longest self-
avoiding $st$-path in $G$ (\cref{cor:estimating-capacitance-query}). This follows from Ref.~\cite{IJ15}, which shows that
given any span program, there is a quantum algorithm that, on input $x$,
outputs an estimate of the witness size of $x$.

\paragraph{New quantum algorithm for connectivity} 
We use this tighter analysis of the negative witness to analyze a new algorithm for graph
connectivity. This problem was first studied in the context of quantum
algorithms by D\"urr, H{\o}yer, Heiligman and Mhalla \cite{durr2006quantum},
who gave an optimal $\widetilde{O}(n^{3/2})$ upper bound on the time
complexity. An optimal span-program-based quantum algorithm was later
presented by {\={A}}ri{\c{n}}{\v{s}} \cite{Arins2016}, whose algorithm also
uses only $O(\log n)$ space. 

Since a graph is connected if and only if every pair of vertices $\{u,v\}$ are
connected, we propose an algorithm that uses the technique of \cite{Nisan:1995:SLC:225058.225101,JK2017} to
convert the conjunction of $\binom{n}{2}$ $st$-connectivity span programs into a
single $st$-connectivity span program: take $n(n-1)/2$ copies of $G(x)$, one for
each pair of distinct vertices $\{u,v\}$ with $u<v$, and call $u$ the source
and $v$ the sink of this graph. Connect these graphs in series, in any
 order, by identifying the sink of one to the source of the next.
Call the source of the first graph $s$, and the sink of the last graph $t$.
See \cref{fig:G'} for an example when $G$ is a triangle. In this way we
have created a graph (which we denote $\mathcal{G}(x)$) that is $st$-connected if and only if
$G(x)$ is connected. In other words, for any $x\in\{0,1\}^{E(G)}$,
$\textsc{conn}_G(x)=st$-$\textsc{conn}_{\mathcal{G}}(x)$.

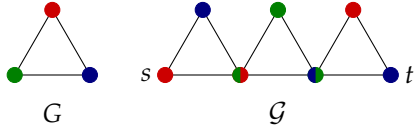
\begin{figure}[h]
\centering
\begin{tikzpicture}

\draw (0,0)--(1,0)--(.5,.866)--(0,0);
\filldraw[black!50!green] (0,0) circle (.1);
\filldraw[black!20!red] (.5,.866) circle (.1);
\filldraw[black!50!blue] (1,0) circle (.1);

\node at (.5,-.5) {$G$};

\draw (2,0)--(5,0)--(4.5,.866)--(4,0)--(3.5,.866)--(3,0)--(2.5,.866)--(2,0);
\filldraw[black!50!green] (3.5,.866) circle (.1);
\filldraw[black!20!red] (4.5,.866) circle (.1);
\filldraw[black!50!blue] (2.5,.866) circle (.1);

\filldraw[black!20!red] (2,0) circle (.1);
\filldraw[black!50!blue] (5,0) circle (.1);

\filldraw[black!20!red] (3,0) circle (.1);
\filldraw[black!50!green] (4,0) circle (.1);

\filldraw[black!50!green] (3,.1) arc (90:270:.1);
\filldraw[black!50!blue] (4,.1) arc (90:270:.1);

\node at (3.5,-.5) {$\mathcal{G}$};
\node at (1.75,0) {$s$};
\node at (5.25,0) {$t$};

\end{tikzpicture}
\caption{The graph $\mathcal{G}$ is $st$-connected if and only of $G$ is connected.}\label{fig:G'}
\end{figure}

For a graph $G(x)$, define the average resistance as 
$R_{\mathrm{avg}}(G(x))=\frac{1}{n(n-1)}\sum_{s,t:s\neq t}R_{s,t}(G(x))$. We
consider the case where we are promised that if $G(x)$ is connected, then
$R_{\mathrm{avg}}(G(x))\leq R$, and if $G(x)$ is not connected it has at least
$\kappa >1$ components. By analyzing the effective resistance and capacitance
of ${\cal G}$, we show that when $G$ is a subgraph of a complete graph,
meaning it has no multi-edges, $\textsc{conn}_{G}$ under this promise can be
solved in query complexity ${O}(n\sqrt{R/\kappa})$ 
(\cref{thm:connectivity-algorithm}), and time complexity $\widetilde{O}(n\sqrt{R/\kappa}\mathsf{U})$,
where $\mathsf{U}$ is the cost of implementing one step of a quantum walk on
$G$ (\cref{cor:connectivity-algorithm}). For the case when $G$ has multi-edges, we get an upper bound of $O(n^{3/4}\sqrt{Rd_{\max}(G)}/\kappa^{1/4})$
on the query complexity, where $d_{\max}(G)$ is the maximum degree of any vertex in the graph.

In the worst case, when $R=n$ and $\kappa=2$, our algorithm for the case when
$G$ is a subgraph of the complete graph achieves the optimal upper bound of
$\widetilde{O}(n^{3/2})$. Like the algorithm of Ref.~\cite{Arins2016}, our
algorithm uses only $O(\log n)$ space. It is also the first connectivity algorithm that applies to the the
case where $G$ is not necessarily the complete graph, although the other
algorithms can likely be adapted to the more general case.

The algorithm of {\={A}}ri{\c{n}}{\v{s}} 
 can be seen as similar to ours, except that rather than connecting copies of
 $G(x)$ for each $\{u,v\}$ pair, his algorithm only considers pairs $\{1,v\}$
 for $v\neq 1$. In contrast, our algorithm is symmetric in the vertex set,
 which makes a detailed analysis more natural.

\paragraph{Alternative quantum algorithms for connectivity} In
\cref{SpecConnectivity}, we present an alternative approach to deciding graph
connectivity. It is based on phase estimation of a particular unitary that is
also used in the  $st$-connectivity span program,
but applied to a different initial state.

We first show that the quantum query complexity of deciding $\textsc{conn}_G$
is $O(\sqrt{nd_{\max}(G)/(\kappa\lambda)})$, when we're promised that if
$G(x)$ is connected, the second smallest eigenvalue of the Laplacian of
$G(x)$, $\lambda_2(G(x))$, is at least $\lambda$, and otherwise, $G(x)$ has at
least $\kappa>1$ connected components (\cref{cor:query}). In the unweighted worst case, 
$\lambda_2(G(x))\geq 2/n^2$ and $d_{\max}=n-1$, which gives a sub-optimal
$O(n^2)$ algorithm. However, for some classes of inputs, this algorithm
performs better than our first algorithm. Neglecting constants, and using the
fact that $R_{\mathrm{avg}}(G(x))\leq 1/\lambda_2(G(x))$, out first algorithm
has query complexity (in the case of no multi-edges)
\begin{equation}
T_1=n\sqrt{R/\kappa}\leq n/\sqrt{\kappa\lambda}
\end{equation}
whereas our second algorithm has query complexity 
\begin{equation}
T_2=\sqrt{nd_{\max}(G)}/\sqrt{\kappa\lambda}.
\end{equation}
When $G$ is a complete graph, our second algorithm can only be worse, since in
that case $d_{\max}(G)=n-1$. However, when $G$ is a Boolean hypercube so that
$d_{\max}=\log n$, our second algorithm may be significantly better.

However, thus far we have only been considering the query complexity of our
second algorithm. This algorithm also requires an initial state of a particular
form, and while this state is independent of the input, it may generally not
be time efficient to produce such a state. We are able to give time-efficient
versions of our second algorithm in two contexts.

First, in \cref{thm:any-G}, we show that for any $G$, under the promise that
if $G(x)$ is connected, then $\lambda_2(G(x))\geq \lambda$, and otherwise
$G(x)$ has at least $\kappa>1$ connected components, we can solve
$\textsc{conn}_G$ in time complexity
\begin{equation}
\widetilde{O}\left(\sqrt{\frac{nd_{\mathrm{avg}}(G)}{\kappa\lambda_2(G)}}\left(\mathsf{S}+\sqrt{\frac{d_{\mathrm{max}}(G)}{\lambda}}\mathsf{U}\right)\right),
\end{equation}
where $\mathsf{U}$ is the complexity of implementing a step of a quantum walk
on $G$, $\mathsf{S}$ is the cost generating a quantum state corresponding
to the stationary distribution of a random walk on $G$, and $d_{\mathrm{avg}}(G)$ is the average degree of the vertices of $G$. This time complexity might generally be
significantly worse than the query complexity, but has the advantage of
applying to all $G$.

Second, in \cref{thm:Cayley-graph}, we give an improved time complexity for
the case where $G$ is a Cayley graph of degree $d$, showing an upper bound of
\begin{equation}
\widetilde{O}\left(\sqrt{\frac{nd}{\kappa\lambda}}\mathsf{U}+\sqrt{\frac{nd}{\kappa\lambda_2(G)}}\mathsf{\Lambda} \right),
\end{equation}
where, as above, $\mathsf{U}$ is the cost of implementing a step of a quantum
walk on $G$, and $\mathsf{\Lambda}$ is the cost of computing the eigenvalues
of $G$; that is, given $g\in V(G)$, computing $\lambda_g$, the corresponding
eigenvalue. For example, this gives an upper bound of
$\widetilde{O}(n/\sqrt{\lambda\kappa})$ when $G$ is a complete graph, and
$\widetilde{O}(\sqrt{n/(\lambda\kappa)})$ when $G$ is a Boolean hypercube.

We remark that our alternative connectivity algorithms apply for any choice of
edge weights on $G$: $d_{\mathrm{avg}}(G)$ and $d_{\max}(G)$ should be
interpreted as the average and maximum \emph{weighted} degrees in $G$, and
$\lambda_2(G)$ and $\lambda_2(G(x))$ the second-smallest eigenvalue of the
\emph{weighted} Laplacian of $G$ and $G(x)$ respectively. The choice of
weights may also impact the costs $\mathsf{U}$ and $\mathsf{S}$.

\paragraph{Estimating the algebraic connectivity}
We give an algorithm to estimate the \emph{algebraic connectivity} of $G(x)$, $\lambda_2(G(x))$, when
$G$ is a complete graph. The algebraic connectivity is closely related to the inverse of the mixing time, which is known to be small for many interesting families of graphs such as expander graphs. We give a protocol that with probability
at least 2/3 outputs an estimate of $\lambda_2(G(x))$ up to multiplicative error
$\varepsilon$ in time complexity
$\widetilde{O}\left(\frac{1}{\varepsilon}\frac{n}{\sqrt{\lambda_2(G(x))}}\right)$ (\cref{thm:Connectivity-estimation}).

\subsection{Open Problems}
Our work suggests several directions for new research. Since $st$-connectivity is fairly ubiquitous, it seems that our approach may, in turn, help analyze applications of $st$-connectivity. Additionally, we provide two algorithms for deciding connectivity, in \cref{sec:connectivity} and \cref{SpecConnectivity}. At least naively, it seems like our two algorithms are incomparable, even though they are based on similar unitaries. It would be worthwhile to understand whether the two approaches are fundamentally different. Another open question is to determine how to set the weights of edges in our graphs; these weights can have a significant effect on query complexity. Finally, it would be interesting to see whether one can extend our algorithm for estimating algebraic connectivity to accept more general parent graphs than the complete graph.

\subsection{Organization}

The rest of this paper is organized as follows. In \cref{sec:Preliminaries},
we introduce the necessary background on which we build our results, including
basic notation (\cref{sec:LinAlg}), graph theory (\cref{sec:graphPrelim}),
quantum algorithms (\cref{sec:alg}) and span programs (\cref{sec:prelimSpan}).
In \cref{sec:ConnectSP}, we show that the negative witness size of an
$st$-connectivity span program is the effective capacitance of the graph, then
in \cref{sec:Applications}, we give two applications of this observation: The
first is a quantum algorithm for estimating the effective capacitance of a
graph (\cref{sec:estimatingCapacitance}); and the second is our first quantum
algorithm for deciding connectivity, as a composition of $st$-connectivity
span programs (\cref{sec:connectivity}). Finally, in \cref{SpecConnectivity},
we give our second algorithm for deciding connectivity, based on estimating
the second-smallest eigenvalue of the Laplacian, and also give an algorithm
for estimating the algebraic connectivity of a graph.

\section{Preliminaries}\label{sec:Preliminaries}
\subsection{Linear Algebra Notation}\label{sec:LinAlg}

For a subspace $V$ of some inner product space, we let $\Pi_V$ denote the
orthogonal projector onto $V$.

For a linear operator $A$, we will let $\sigma_{\min}(A)$ denote its smallest
non-zero singular value, and $\sigma_{\max}(A)$ its largest singular value. We
let $\ker A$ denote the kernel of $A$, $\mathrm{row}(A)$ denote the rowspace
of $A$, and $\mathrm{col}(A)$ the columnspace of $A$. We let $A^+$ denote the
Moore-Penrose pseudoinverse of $A$. If $A$ has singular value decomposition
$A=\sum_i\sigma_i\ket{a_i}\bra{b_i}$, (for left singular vectors
$\{\ket{a_i}\}$ and right singular vectors $\{\ket{b_i}\}$) then $A^+=\sum_i
1/\sigma_i \ket{b_i}\bra{a_i}$. Then we have $AA^+=\Pi_{\mathrm{col}(A)}$ and
$A^+A=\Pi_{\mathrm{row}(A)}$.

For a unitary $U$ with eigenvalues $e^{i\theta_1},\dots,e^{i\theta_N}$ for $\theta_1,\dots,\theta_N\in (-\pi,\pi]$, let
$\Delta(U)=\min\{|\theta_i|:\theta_i\neq 0\}$ denote the \emph{phase gap} of
$U$.

\subsection{Graph Theory}\label{sec:graphPrelim}

\paragraph{Multigraphs} We will consider multigraphs, which may have multiple edges between a pair of vertices. Thus, to each edge, we associate a unique identifying
label $\ell$. We refer to each edge in the graph using its endpoints and the
label $\edgeL$, as, for example: $(\{u,v\},\edgeL)$. The label $\edgeL$
uniquely specifies the edge, but we include the endpoints for convenience. Let
$\overrightarrow{E}(G)=\{(u,v,\edgeL ):(\{u,v\},\edgeL )\in E(G)\}$ be the
directed edges of~$G$. Furthermore, for any set of edges $E$, we let
$\overrightarrow{E} =\{(u,v,\edgeL ):(\{u,v\},\edgeL )\in E\}$ represent the
corresponding set of directed edges. We will sometimes write $(u,v,\ell)$ for an undirected edge, but when talking about undirected edges, we have $(u,v,\ell) = (v,u,\ell)$.

We will be concerned with certain subgraphs of a graph $G$, associated with
bit strings of length $N$. We denote by $G(x)$ the subgraph associated with
the string $x\in \{0,1\}^N$. In particular, each edge in $G$ is associated
with a variable $x_i$ or its negation $\overline{x_i}$, called a
\emph{literal}, and is included in $G(x)$ if and only if the associated
literal evaluates to 1. Here $x_i$ is the $i\tth$ bit of $x$. For example, an
edge $(u,v,\edgeL)$ associated with the literal $\overline{x_i}$ is in $G(x)$
if and only if $x_i$ takes value $0$, as in \cref{fig:Gx}. Precisely how this
association of edges and literals is chosen depends on the problem of
interest, so we will leave the description implicit, and often assume for simplicity that there is a one-to-one mapping between the edges and positive literals.

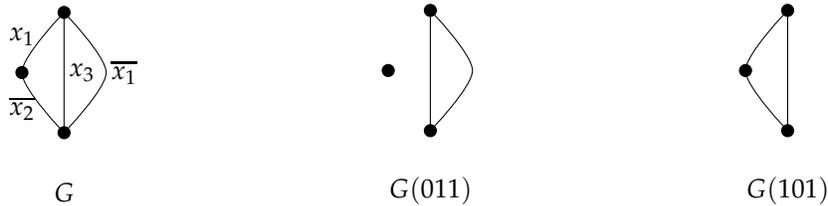
\begin{figure}[ht]
\centering
\begin{tikzpicture}[scale = .95]
\node at (0,0) {\begin{tikzpicture}[scale=.8]
\filldraw (0,1) circle (.1);
\filldraw (0,-1) circle (.1);
\filldraw (-.7,0) circle (.1);

\draw (0,1)--(0,-1);
\draw plot [smooth] coordinates{(0,1) (-.7,0) (0,-1)};
\draw plot [smooth] coordinates{(0,1) (.7,0) (0,-1)};

\node at (-.7,.6) {$x_1$};
\node at (-.7,-.6) {$\overline{x_2}$};
\node at (.3,0) {$x_3$};
\node at (1,0) {$\overline{x_1}$};

\node at (0,-2) {$G$};
\end{tikzpicture}};
\node at (5,0) {\begin{tikzpicture}[scale=.8]
\filldraw (0,1) circle (.1);
\filldraw (0,-1) circle (.1);
\filldraw (-.7,0) circle (.1);

\draw (0,1)--(0,-1);
\draw plot [smooth] coordinates{(0,1) (.7,0) (0,-1)};

\node at (0,-2) {$G(011)$};
\end{tikzpicture}};
\node at (10,0) {\begin{tikzpicture}[scale=.8]
\filldraw (0,1) circle (.1);
\filldraw (0,-1) circle (.1);
\filldraw (-.7,0) circle (.1);

\draw (0,1)--(0,-1);
\draw plot [smooth] coordinates{(0,1) (-.7,0) (0,-1)};

\node at (0,-2) {$G(101)$};
\end{tikzpicture}};
\end{tikzpicture}
\caption{Example of how each edge in $G$ is associated with a bit of $x$ and a
value of that bit. For edges labeled by $x_i$ we include the
edge in $G(x)$ if $x_i=1$, while for edged labeled by $\overline{x_i}$, we
include the edge in $G(x)$ if $x_i=0$.}\label{fig:Gx}
\end{figure}

\paragraph{Networks} A \emph{network} $\sop N=(G,c)$ consists of a graph $G$
combined with a positive real-valued \emph{weight} function
$\Ohm:E(G)\longrightarrow\mathbb{R}^+$. Since $c$ is a map on undirected
edges, we can easily extend it to a map on directed edges such that $c(u,v,\ell)
= c(v,u,\ell)$, and we overload our notation accordingly. We will often assume
that some $c$ is implicit for a graph $G$ and let
\begin{equation}
{\cal A}_G=\sum_{(u,v,\edgeL )\in E(G)}c(u,v,\ell)(\ket{u}\bra{v}+\ket{v}\bra{u})
\end{equation}
denote its weighted adjacency matrix. Note that ${\cal A}_G$ only depends on
the total weight of edges from $u$ to $v$, and is independent of the number of
edges across which this weight is distributed. Let
$d_G(u)=\sum_{v,\edgeL:(u,v,\edgeL)\in E(G)}c(u,v,\edgeL)$ denote the weighted
degree of $u$ in $G$, under the implicit weight function $c$, and let
$d_{\max}(G)=\max_{u\in V(G)}d_G(u)$. Let
\begin{equation}
{\cal D}_G = \sum_{u\in V(G)}d_G(u)\ket{u}\bra{u}
\end{equation}
denote the weighted degree matrix, and let 
\begin{equation}
L_G = {\cal D}_G - {\cal A}_G
\end{equation}
denote the Laplacian of $G$. The Laplacian is always positive semidefinite, so
its eigenvalues are real and non-negative. For $\ket{\mu}=\sum_{u\in
V(G)}\ket{u}$, it is always the case that $L_G\ket{\mu} = 0$, so the smallest
eigenvalue of $L_G$ is 0. Let $\lambda_2(G)$ denote the second smallest
eigenvalue of $L_G$, including multiplicity. This value is called the
\emph{algebraic connectivity} or the \emph{Fiedler value} of $G$, and it is
non-zero if and only if $G$ is connected.

\paragraph{Electric networks} Consider a graph $G$ with specially labeled
vertices $s$ and $t$ that are connected in $G$. One can consider a fluid that
enters a graph $G$ at $s$, flows along the edges of the graph, and exits the
graph at $t$. The fluid can spread out along some number of the $st$-paths in
$G$. An \emph{$st$-flow} is any linear combination of $st$-paths. More
precisely:
\begin{definition}[Unit $st$-flow]\label{def:st-flow}
Let $G$ be an undirected graph with $s,t\in V(G)$, and $s$ and $t$ connected.
Then a \emph{unit $st$-flow} on $G$ is a function
$\theta:\overrightarrow{E}(G)\rightarrow\mathbb{R}$ such that:
\begin{enumerate}
\item For all $(u,v,\edgeL)\in \overrightarrow{E}(G)$, $\theta(u,v,\edgeL )=-\theta(v,u,\edgeL)$;
\item $\sum_{v,\edgeL:(s,v,\edgeL)\in \overrightarrow{E}(G)}\theta(s,v,\edgeL)=\sum_{v,\edgeL:(v,t,\edgeL)\in \overrightarrow{E}(G)}\theta(v,t,\edgeL)=1$; and 
\item for all $u\in V(G)\setminus\{s,t\}$, $\sum_{v,\edgeL:(u,v,\edgeL)\in \overrightarrow{E}(G)}\theta(u,v,\edgeL)=0$. 
\end{enumerate}
\end{definition}

\begin{definition}[Unit Flow Energy]\label{def:unitFlowEnergy}
Given a graph $G$ with implicit weighting $c$ and a unit $st$-flow $\theta$ on $G(x)$, the \emph{unit flow energy} of $\theta$ on $E'\subseteq E(G(x))$, is
\begin{align}
J_{E'}(\theta)=\frac{1}{2}\sum_{e\in {\overrightarrow{E'}}}\frac{\theta(e)^2}{c(e)}.
\end{align}

\end{definition}

\begin{definition}[Effective resistance] Let $G$ be a graph with implicit weighting $c$ and $s,t\in V(G)$.
If $s$ and $t$ are connected in $G(x)$, the \emph{effective resistance} of
$G(x)$ between $s$ and $t$ is $R_{s,t}(G(x)) = \min_{\theta}
J_{E(G(x))}(\theta)$, where $\theta$ runs over all unit $st$-unit flows of
$G(x)$. If $s$ and $t$ are not connected in $G(x)$, $R_{s,t}(G(x))=\infty.$
\label{def:effRes}
\end{definition}

Intuitively, $R_{s,t}$ characterizes ``how connected'' the vertices $s$ and
$t$ are in a network. The more, shorter paths connecting $s$ and $t$, and the
more weight on those paths, the smaller the effective resistance.
 
The effective resistance has many applications. For example,
$R_{s,t}(G)\left(\sum_{e\in E(G)}c(e)\right)$ is equal to the \emph{commute
time} between $s$ and $t$, or the expected time a random walker starting from
$s$ takes to reach $t$ and then return to $s$
\cite{CRRST96}. If $\sop N=(G,c)$ models an electrical
network in which each edge $e$ of $G$ is a $1/c(e)$-valued resistor and a
potential difference is applied between $s$ and $t$, then $R_{s,t}(\sop N)$
corresponds to the resistance of the network, which determines the ratio of
current to voltage in the circuit (see \cite{DS84}). Thus, the values $c(e)$
can be interpreted as conductances.

For a connected graph $G$, we can define the \emph{average resistance} by:
$$R_{\mathrm{avg}}(G):=\frac{1}{n(n-1)}\sum_{s,t\in V:s\neq t}R_{s,t}(G).$$

Now that we have a measure of the connectedness of $s$ and $t$ in a graph $G$,
we next introduce a measure of how disconnected $s$ and $t$ are, in the case that
we are considering a subgraph $G(x)$ of $G$ where $s$ and $t$ are not
connected.

\begin{definition}[Unit $st$-potential]
Let $G$ be an undirected weighted graph with $s,t\in V(G)$, and $s$ and $t$
connected. For $G(x)$ such that $s$ and $t$ are not connected, a \emph{unit
$st$-potential} on $G(x)$ is a function $\sop V:V(G)\rightarrow\mathbb{R}^+$
such that $\sop V(s)=1$ and $\sop V(t)=0$ and $\sop V(u)=\sop V(v)$ if
$(u,v,\edgeL)\in E(G(x)).$
\label{def:dualPotential}
\end{definition}

Note that this is a different definition from the typical potential function.
Usually, if we have a flow from a vertex $s$ to a vertex $t$, we define the
potential difference between $u$ and $v$ for an edge $(u,v,\edgeL)$ to be the
amount of flow across that edge divided by the weight of the edge. In our
definition, the potential difference across all edges in $E(G(x))$ is zero,
and we have potential difference across edges that are in $E(G)\setminus
E(G(x))$.

A unit $st$-potential is a witness of the disconnectedness of $s$ and $t$ in
$G(x)$, in the sense that it is a generalization of the notion of an $st$-cut.
(An $st$-cut is a unit potential that only takes values $0$ and $1$.)

\begin{definition}[Unit Potential Energy]\label{def:unitPotentEnergy}
Given a graph $G$ with implicit weighting $c$ and a unit $st$-potential $\sop
V$ on $G(x)$, the \emph{unit potential energy} of $\sop V$ on $E'\subseteq
E(G)$ is defined
\begin{align}
\sop J_{E'}(\sop V)=\frac{1}{2}\sum_{(u,v,\edgeL )\in {\overrightarrow{E'}}}(\sop V(u)-\sop V(v))^2c(u,v,\edgeL ).
\end{align}
\end{definition}

\begin{definition}[Effective capacitance] Let $G$ be a graph with implicit weighting $c$ and $s,t\in V(G)$.
If $s$ and $t$ are not connected in $G(x)$, the \emph{effective capacitance}
between $s$ and $t$ of $G(x)$ is $C_{s,t}(G(x)) = \min_{\sop V} \sop
J_{E(G)}(\sop V)$, where $\sop V$ runs over all unit $st$-potentials on
$G(x)$. If $s$ and $t$ are connected, $C_{s,t}(G(x))=\infty.$
\label{def:effCap}
\end{definition}

In physics, capacitance is a measure of how well a system stores electric
charge. The simplest capacitor is a set of separated metal plates at a fixed
distance. To see how a capacitor works, we imagine one terminal of a battery attached to each plate. The battery acts as a sort of pump that moves negative charges from one plate to the other, against their natural tendencies. Because like charges repel, as the plates become increasingly polarized, it
requires more energy to move additional charges. In other
words, as additional charge accumulates, the voltage difference between the plates increases and, correspondingly, the voltage that it takes to overcome the differential increases. At equilibrium, the voltage difference between the plates will be equal to the voltage difference of the terminals of the battery. The ratio of the amount of charge moved to the voltage difference
created between the plates is a property of the capacitor itself and depends only on its geometry. This ratio is called its effective capacitance.

Now consider a capacitor that corresponds to the graph $G(x)$ in which a
$0$-resistance wire is connected between vertices whenever there is an edge{}
in $G(x)$, and a $c(e)$-unit capacitor is connected between vertices whenever
there is an edge $e\in E(G)\setminus E(G(x))$. If $s$ and $t$ are not
connected in $G(x)$, it is as though $s$ and $t$ are on separate ``plates''
(with some complicated geometry) that can accumulate charge relative to each
other. Then the effective capacitance given in \cref{def:effCap} is precisely
the ratio of charge (accumulated on the plates corresponding to $s$ and
$t$) to voltage (on those plates) that is achieved when electrical energy is
stored in this configuration. We make the connection between \cref{def:effCap}
and the standard definition of effective capacitance, as well as effective
conductance, more explicitly in \cref{ap:effCapDef}.

\cref{def:effRes,def:effCap} may seem unwieldy for actually calculating the
effective resistance and effective capacitance, so we now recall that when
calculating effective resistance, $R_{s,t}$ (respectively effective
capacitance $C_{s,t}$), one can use the rule that for edges in series, or more
generally, graphs connected in series, resistances add (resp. inverse
capacitances add). Edges in parallel, or more generally, graphs connected in
parallel, follow the rule that inverse resistances add (resp. capacitances
add). That is:

\begin{claim}\label{claim:parallel_series}
Let two networks $(G_1,c_1)$ and $(G_2,c_2)$ each have connected nodes $s$ and
$t$. Let $G(x_1)$ and $G(x_2)$ be subgraphs of $G_1$ and $G_2$ respectively.
Then we consider a new graph $G$ by identifying the $s$ nodes and the $t$
nodes of $G_1$ and $G_2$ (i.e. connecting the graphs in parallel) and define
$c:E(G)\rightarrow\mathbb{R}^+$ by $c(e)=c_1(e)$ if $e\in E(G_1)$ and
$c(e)=c_2(e)$ if $e\in E(G_2)$. Similarly, we set $G(x)$ to be the subgraph of
$G$ that includes the corresponding edges $e$ such that $e\in E(G_1(x_1))$ or
$e\in E(G_2(x_2))$. Then
\begin{align}\label{eq:parallel}
\frac{1}{R_{s,t}(
G(x))}=\frac{1}{R_{s,t}(
G_1(x))}+\frac{1}{R_{s,t}(
G_2(x))},\qquad
C_{s,t}(G(x))=C_{s,t}(G_1(x_1))+C_{s,t}(G_2(x_2))
\end{align}
If we create a new graph $G$ by identifying the $t$ node of $G_1$ with the $s$
node of $G_2$, relabeling this node $v\not\in\{s,t\}$ (i.e. connecting the
graphs in series) and define $c$ and $G(x)$ as before, then
\begin{align}\label{eq:series}
R_{s,t}(G(x))=R_{s,t}(G_1(x))+R_{s,t}(G_2(x)),\qquad
\frac{1}{C_{s,t}(G(x))}=\frac{1}{C_{s,t}(G_1(x_1))}+\frac{1}{C_{s,t}(G_2(x_2))}.
\end{align}
\end{claim}

\subsection{Quantum Algorithms}\label{sec:alg}

We will consider problems parametrized by a parent graph $G$, by which we more
precisely mean a family of graphs $\{G_n\}_{n\in\mathbb{N}}$ where $G_n$ is a
graph on $n$ vertices. We will generally drop the subscript $n$. 

A graph is connected if there is a path between every pair of vertices. For a
family of graphs $G$, and a set $X\subset \{0,1\}^{N}$, let
\textsc{conn}$_{G,X}$ denote the \emph{connectivity problem}, defined by
$\textsc{conn}_{G,X}(x)=1$ if $G(x)$ is connected (see \cref{sec:graphPrelim} for description of $G(x))$, and
$\textsc{conn}_{G,X}(x)=0$ if $G(x)$ is not connected, for all $x\in X$.

Similarly, for $s,t\in V(G)$, defined $st$-\textsc{conn}$_{G,X}$ by
$st$-\textsc{conn}$_{G,X}(x)=1$ if there is a path from $s$ to $t$ in $G(x)$,
and $st$-\textsc{conn}$_{G,X}(x)=0$ otherwise, for all $x\in X$.

We will consider \textsc{conn} and $st$-\textsc{conn} in the edge-query input model, meaning that we
have access to a standard quantum oracle $O_x$, defined
$O_x\ket{i}\ket{b}=\ket{i}\ket{b\oplus x_i}$, where $x_i$ is the $i\tth$ bit
of $i$. Since every edge of $G$ is associated with an input variable, as
described in \cref{sec:graphPrelim}, for any edge in $G$, we can check if it
is also present in $G(x)$ using one query to $O_x$.

Let $f_N:X_N\rightarrow\{0,1\}$, $X_N\subseteq [q]^N$, for $N\in\mathbb{N}$, be
a family of functions. An algorithm \emph{decides $f$ with bounded error} if
for any $x\in X_N$, the algorithm outputs $f(x)$ with probability at least
$2/3$. Let $f_N:X_N\rightarrow\mathbb{R}_{\geq 0}$, $X_N\subseteq [q]^N$, for
$N\in\mathbb{N}$, be a family of functions. An algorithm \emph{estimates $f$
to relative accuracy $\eps$ with bounded error} if for any $x\in X_N$, on
input $x$, the algorithm outputs $\tilde{f}$ such that $|\tilde{f}-f(x)|\leq
\eps f(x)$ with probability at least $2/3$. Since all algorithms discussed
here will have bounded error, we will omit this description. We will generally
talk about a function $f:X\rightarrow\{0,1\}$, leaving the parametrization
over $n$ implicit. 

In the remainder of this section, we describe the quantum algorithmic building
blocks of the algorithms introduced in \cref{SpecConnectivity}.

\begin{theorem}[Phase Estimation \cite{kit95,CEMM98}]\label{thm:phase-estimation}
Let $U$ be a unitary with eigenvectors $\ket{\theta}$ satisfying
$U\ket{\theta}=e^{i\theta}\ket{\theta}$ and assume $\theta\in[-\pi,\pi]$. For
any $\Theta\in(0,\pi)$ and $\varepsilon\in(0,1)$, there exists a quantum
algorithm that makes $O(\frac{1}{\Theta}\log \frac{1}{\varepsilon})$ calls to
$U$ and, on input $\ket{\theta_j}$ outputs a state $\ket{\theta_j}\ket{w}$
such that if $\theta_j=0$, then $\ket{w}=\ket{0}$ and if $|\theta_j|\geq
\Theta$, $|\braket{0}{w}|^2\leq \varepsilon$. If $U$ acts on $s$ qubits, the
algorithm uses $O(s+\log\frac{1}{\Theta})$ space.
\end{theorem}

We will use the following corollary of \cref{thm:phase-estimation}, which is a slight generalization of an algorithm introduced in \cite{CKS17}, also called Gapped Phase Estimation.

\begin{theorem}[Gapped Phase Estimation]\label{thm:gapped-phase-estimation}
Let $U$ be a unitary with eigenvectors $\ket{\theta}$ satisfying
$U\ket{\theta}=e^{i\pi\theta}\ket{\theta}$ and assume $\theta\in[-1,1]$. Let
$\varphi\in(0,1)$, let $\eps>0$ and let $\delta\in(0,1-\varphi]$. Then, there exists a unitary procedure
$GPE(\varphi,\eps,\delta)$ making $O(\varphi^{-1}\log \eps^{-1})$ queries to $U$
that on input $\ket{0}_C\ket{0}_P\ket{\theta}$ prepares a state $(\beta_0\ket{0}_C\ket{\gamma_0}_P+\beta_1\ket{1}_C\ket{\gamma_1}_P)\ket{\theta}$ where
$\ket{\gamma_0}$ and $\ket{\gamma_1}$ are some unit vectors, $\beta_0^2+\beta_1^2=1$ and such that
\begin{itemize}
\item if $|\theta|\leq \delta$, then $|\beta_1|\leq\eps$, and
\item if $\delta+\varphi\leq |\theta|$, then $|\beta_0|\leq\eps$.
\end{itemize}
The registers $C$ and $P$ have 1 and $O(\log(\varphi^{-1})\log(\eps^{-1}))$ qubits respectively. And in addition to the queries to $U$,  the algorithm uses $O\left(\log\frac{1}{\eps}(\log^2\frac{1}{\varphi}+\log\log\frac{1}{\eps})\right)$ elementary gates.
\end{theorem}
\begin{proof}
Standard phase estimation \cite{kit95,CEMM98} (but see, in particular \cite[Appendix C]{CEMM98}) on input $\ket{\theta}$ with precision $\varphi/4$ prepares a state $\ket{\tilde\theta}\ket{\theta}$ such that upon measuring $\ket{\tilde\theta}$, with probability at least $c$ for some $c>1/2$, we measure some $\bar\theta$ that is within $\varphi/2$ of $\theta$, meaning that if $|\theta|\leq \delta$, then $|\bar\theta|<\delta+\varphi/2$, and if $|\theta|\geq \delta+\varphi$, then $|\bar\theta|> \delta+\varphi/2$.

Let $k=c'\log\frac{1}{\eps^2}$ for some constant $c'$. Repeating phase estimation $k$ times produces a state $\ket{\tilde{\theta}}_{P_1}\dots\ket{\tilde{\theta}}_{P_k}$ such that if we were to measure the state, we would obtain a string of estimates $\bar{\theta}_1,\dots,\bar{\theta}_{k}$. If $|\theta|\leq \delta$, each estimate would have absolute value less than $\delta+\varphi/2$ with probability at least $c$, and if $|\theta|\geq \delta+\varphi$, each estimate would have absolute value greater than $\delta+\varphi/2$ with probability at least $c$. 

Instead of measuring, assume that we have extra registers $\ket{0}_{C_1}\dots\ket{0}_{C_k}$. Apply to every pair of registers $P_iC_i$ the unitary that maps $\ket{\bar\theta}_{P_i}\ket{0}_{C_i}$ to $\ket{\bar\theta}_{P_i}\ket{0}_{C_i}$ if $|\bar\theta|\leq \delta+\varphi/2$, and to $\ket{\bar \theta}_{P_i}\ket{1}_{C_i}$ otherwise. This unitary can be done with $O(\log^2\frac{1}{\varphi})$ elementary gates. Then we unitarily do majority voting of the values of registers $C_1,\dots,C_k$ and encode the result in a register $C$. This step requires only $O(\log\frac{1}{\eps}\log\log\frac{1}{\eps})$ elementary gates. By a standard Chernoff bound, for appropriately chosen $c'$, this majority will be 0 with amplitude at least $\sqrt{1-\eps^2}$ 
if $|\theta|\leq \delta$, and will be 1 with amplitude at least $\sqrt{1-\eps^2}$ if $|\theta|\geq \delta+\varphi$. 

Grouping all registers $P_1,\dots P_k,C_1,\dots,C_k$ into one label $P$ we have that the state produced will be $(\beta_0\ket{0}_C\ket{\gamma_0}_P+\beta_1\ket{1}_C\ket{\gamma_1}_P)\ket{\theta}$ for some unit vectors $\ket{\gamma_0}$ and $\ket{\gamma_1}$, with $|\beta_1|\leq \eps$ whenever $|\theta|\leq \delta$, and $|\beta_0|\leq \eps$ whenever $|\theta|\geq \delta+\varphi$. 

The number of elementary gates used in every call to standard phase estimation with precision $\varphi/2$ is $O(\log^2\varphi)$, in addition to $O(\frac{1}{\varphi}\log\frac{1}{\eps})$ calls to $U$, from which the result follows.
\end{proof}

\begin{theorem}[Amplitude Estimation \cite{brassard2002quantum}]\label{thm:amplitude-estimation}
Let $\cal A$ be a quantum algorithm that, on input $x$,  outputs $\sqrt{p(x)}\ket{0}\ket{\Psi_0(x)}+\sqrt{1-p(x)}\ket{1}\ket{\Psi_1(x)}$. Then there exists a quantum algorithm that estimates $p(x)$ to precision $\eps$ using $O\left(\frac{1}{\eps\sqrt{p(x)}}\right)$ calls to $\cal A$. 
\end{theorem}

We will make use of the following corollary (see \cite{IJ15} for a proof).
\begin{corollary}\label{cor:amplitude-estimation}
Let $\cal A$ be a quantum algorithm that outputs $\sqrt{p(x)}\ket{0}\ket{\Psi_0(x)}+\sqrt{1-p(x)}\ket{1}\ket{\Psi_1(x)}$ on input $x$ such that either $p(x)\leq p_0$, or $p(x)\geq p_1$ for $p_1>p_0$. Then there exists a quantum algorithm that decides if $p(x)\leq p_0$ using $O\left(\frac{\sqrt{p_1}}{p_1-p_0}\right)$ calls to $\cal A$.
\end{corollary}

\subsection{Span Programs and Witness Sizes}\label{sec:prelimSpan}

Span programs \cite{KW93} were first introduced to the study of
quantum algorithms by Reichardt and \v{S}palek \cite{RS12}. They have
since proven to be immensely important for designing quantum
algorithms in the query model.

\begin{definition}[Span Program]\label{def:span}
A span program $P=(H,U,\tau,A)$ on $\{0,1\}^N$ is made up of 
{\bf{(I)}} finite-dimensional inner product spaces $H=H_1\oplus \dots \oplus H_N$, and $\{H_{j,b}\subseteq H_j\}_{j\in [N],b\in \{0,1\}}$ such that $H_{j,0}+H_{j,1}=H_j$, {\bf{(II)}} a vector space $U$, {\bf{(III)}} a non-zero \emph{target vector} $\tau\in U$, and {\bf{(IV)}} a linear operator $A:H\rightarrow U$.
For every string $x\in \{0,1\}^N$, we associate the subspace $H(x):=H_{1,x_1}\oplus \dots\oplus H_{N,x_N}$, and an operator $A(x):=A\Pi_{H(x)}$. 
\end{definition}


\begin{definition}[Positive and Negative Witness]\label{def:posNegWit}
Let $P$ be a span program on $\{0,1\}^N$ and let $x$ be a string $x\in \{0,1\}^N$.
Then we call $\ket{w}$ a \emph{positive witness for $x$ in $P$} if
$\ket{w}\in H(x)$, and $A\ket{w}=\tau$. We define the \emph{positive
witness size of $x$} as:
\begin{equation}
w_+(x,P)=w_+(x)=\min\{\norm{\ket{w}}^2: \ket{w}\in H(x),A\ket{w}=\tau\},
\end{equation}
if there exists a positive witness for $x$, and $w_+(x)=\infty$ otherwise.

Let $\mathcal{L}(U,\mathbb{R})$ denote the set of linear maps from $U$ to $\mathbb{R}.$ We call a linear map $\omega\in\mathcal{L}(U,\mathbb{R})$ a \emph{negative
witness for $x$ in $P$} if $\omega A\Pi_{H(x)} = 0$ and $\omega\tau =
1$. We define the \emph{negative witness size of $x$} as:
\begin{equation}
w_-(x,P)=w_-(x)=\min\{\norm{\omega A}^2:{\omega\in \mathcal{L}(U,\mathbb{R}), 
\omega A\Pi_{H(x)}=0, \omega\tau=1}\},
\end{equation}
if there exists a negative witness, and $w_-(x)=\infty$ otherwise.

If $w_+(x)$ is finite, we say that $x$ is \emph{positive} (wrt.\ $P$), 
and if $w_-(x)$ is finite, we say that $x$ is \emph{negative}. We let 
$P_1$ denote the set of positive inputs, and $P_0$ the set of negative 
inputs for $P$.  
\end{definition}

For a function $f:X\rightarrow \{0,1\}$, with $X\subseteq\{0,1\}^N$, we say
 $P$ \emph{decides} $f$ if $f^{-1}(0)\subseteq P_0$ and
$f^{-1}(1)\subseteq P_1$. Given a span program $P$ that decides $f$, one can use it to design a quantum algorithm whose output is $f(x)$ (with high probability), 
given access to the input $x\in X$ via queries of the form
$\mathcal{O}_x:\ket{i,b}\mapsto \ket{i,b\oplus x_i}$.

The following theorem is due to \cite{Rei09} (see \cite{IJ15} for a version with similar notation).
\begin{theorem}\label{thm:span-decision}
Let $U(P,x)=(2\Pi_{\ker A}-I)(2\Pi_{H(x)}-I)$. 
Fix $X\subseteq\{0,1\}^N$ and $f:X\rightarrow\{0,1\}$, and let $P$ be a span program on $\{0,1\}^N$ that decides $f$. 
Let $W_+(f,P)=\max_{x\in f^{-1}(1)}w_+(x,P)$ 
and $W_-(f,P)=\max_{x\in f^{-1}(0)}w_-(x,P)$. 
Then there is a bounded error quantum algorithm that decides $f$ by making $O(\sqrt{W_+(f,P)W_-(f,P)})$ calls to $U(P,x)$, and elementary gates. In particular, this algorithm has quantum query complexity $O(\sqrt{W_+(f,P)W_-(f,P)})$.
\end{theorem}

Ref.~\cite{IJ15} defines the \emph{approximate positive and negative witness
sizes}, $\tilde{w}_+(x,P)$ and $\tilde{w}_-(x,P)$. These are similar to the
positive and negative witness sizes, but with the conditions $\ket{w}\in H(x)$
and $\omega A\Pi_{H(x)}=0$ relaxed.

\begin{definition}[Approximate Positive Witness]
For any span program $P$ on $\{0,1\}^N$ and $x\in\{0,1\}^N$, we define the
\emph{positive error of $x$ in $P$} as: 
\begin{equation}
e_+(x)=e_+(x,P):=\min\left\{
\norm{\Pi_{H(x)^\bot}\ket{w}}^2:A\ket{w}=\tau\right\}.
\end{equation}
We say $\ket{w}$ is an \emph{approximate positive witness} for
$x$ in $P$ if $\norm{\Pi_{H(x)^\bot}\ket{w}}^2=e_+(x)$ and $A\ket{w}=\tau.$
 We define the \emph{approximate positive witness size} as
\begin{equation}
\tilde{w}_+(x)=\tilde{w}_+(x,P):=\min\left\{\norm{\ket{w}}^2:A\ket{w}=\tau,
\norm{\Pi_{H(x)^\bot}\ket{w}}^2=e_+(x)\right\}.
\end{equation}
\end{definition}

\noindent  If $x\in P_1$, then $e_+(x)=0$. In that case, an
approximate positive witness for $x$ is a positive witness, and
$\tilde w_+(x)=w_+(x)$. For negative inputs,
the positive error is larger than 0. 

We can define a similar notion of approximate negative witnesses (see \cite{IJ15}).

\begin{theorem}[\cite{IJ15}]\label{theorem:span-est}
Let $U(P,x)=(2\Pi_{\ker A}-I)(2\Pi_{H(x)}-I)$.
Fix $X\subseteq\{0,1\}^N$ and $f:X\rightarrow \mathbb{R}_{\geq 0}$. Let $P$ be a
span program on $\{0,1\}^N$ such that for all $x\in X$, 
$f(x)=w_-(x,P)$ and define
$\widetilde{W}_+=\widetilde{W}_+(P)=\max_{x\in X}\tilde{w}_+(x,P)${}.
Then there exists a quantum algorithm that estimates $f$ to accuracy
$\eps$ and that uses $\tO\left(\frac{1}{\eps^{3/2}}\sqrt{w_-(x)\widetilde{W}_+}\right)$
calls to $U(P,x)$ and elementary gates. 
\end{theorem}

\paragraph{A span program for $st$-connectivity} An important example of a
span program is one for $st$-connectivity, first introduced in \cite{KW93}, and
used in \cite{BR12} to give a new
quantum algorithm for $st$-connectivity. We state this span program below,
somewhat generalized to include weighted graphs, and to allow the input to be
specified as a subgraph of some parent graph $G$ that is not necessarily the
complete graph. We allow a string $x\in\{0,1\}^N$ to specify a subgraph $G(x)$
of $G$ in a fairly general way, as described in \cref{sec:graphPrelim}. In
particular, for $i\in [N]$, let $\overrightarrow{E}_{i,1}\subseteq
\overrightarrow{E}(G)$ denote the set of (directed) edges associated with the
literal $x_i$, and $\overrightarrow{E}_{i,0}$ the set of edges associated with
the literal $\overline{x_i}$. Note that if $(u,v,\edgeL)\in
\overrightarrow{E}_{i,b}$ then we must also have $(v,u,\edgeL)\in
\overrightarrow{E}_{i,b}$, since $G(x)$ is an undirected graph. We assume $G$
has some implicit weighting function $c$.

Then we refer to the following span program as $P_{G}$:
\begin{align}
\forall i\in [N], b\in\{0,1\}:\; H_{i,b} &= \mathrm{span}\{\ket{e}:e\in \overrightarrow{E}_{i,b}\}\nonumber\\
U &= \mathrm{span}\{\ket{v}:v\in V(G)\}\nonumber\\
\tau &= \ket{s} - \ket{t}\nonumber\\
\forall e=(u,v,\ell)\in\overrightarrow{E}(G):\; A\ket{u,v,\edgeL} &= \sqrt{c(u,v,\edgeL)}(\ket{u} - \ket{v})
\label{eq:st-conn-span-program}
\end{align}
One can check that if $s$ and $t$ are connected, then if $\ket{w}$ represents
a weighted $st$-path or linear combination of weighted $st$-paths in $G(x)$, then $\ket{w}$ is
a positive witness for $x$. Furthermore, this is the only possibility for a
positive witness, so $x$ is a positive input for $P_{G}$ if and only if
$G(x)$ is $st$-connected, and in particular, $w_+(x,P_G)=\frac{1}{2}R_{s,t}(G(x))$ \cite{BR12}. Since the weights $c(e)$ are positive, the set of positive inputs of $P_{G}$ are independent of the choice of $c$, however, the witness
sizes will depend on $c$.


\section{Effective Capacitance and {\it st}-connectivity}\label{sec:ConnectSP}

In this section, we will prove the following theorem:
\begin{theorem}\label{thm:negwit-capacitance}
Let $P_G$ be the span program in Eq.~\eqref{eq:st-conn-span-program}. Then for any $x\in\{0,1\}^N$, $w_-(x,P_G)=2C_{s,t}(G(x))$.
\end{theorem}

Previously, the negative witness size of $P_{G}$ was
characterized by the size of a cut~\cite{RS12} or, in planar graphs, the
effective resistance of a graph related to the planar dual of $G(x)$
\cite{JK2017}.

We will prove \cref{thm:negwit-capacitance} shortly, but first, we mention the following corollary:

\begin{corollary}\label{thm:stconn}
Let $G$ be a multigraph with $s,t\in V(G)$. Then for \emph{any} choice of (non-negative, real-valued) implicit weight function, the bounded error quantum query complexity of evaluating $st$-\textsc{conn}$_{G,X}$ is
\begin{align}\label{eq:ourBound}
O\left(\sqrt{\max_{\substack{x\in X \\ st\textsc{-conn}_{G,X}(x)=1}}R_{s,t}(G(x))\times\max_{\substack{x\in X \\ st\textsc{-conn}_{G,X}(x)=0}}C_{s,t}(G(x))}\right).
\end{align}
\end{corollary}
\begin{proof}
This follows from \cref{thm:negwit-capacitance} and the fact that $w_+(x,P_G)=\frac{1}{2}R_{s,t}(G(x))$, which is proven in \cite{BR12}, and generalized to the weighted case in \cite{JK2017}. Then \cref{thm:span-decision} gives the result.
\end{proof}

We emphasize that \cref{thm:stconn} holds for $R_{s,t}$ and $C_{s,t}$ defined with respect to any weight function, some of which may give a significantly better complexity for solving this problem.

\vskip10pt
We are now ready to prove \cref{thm:negwit-capacitance}, the main result of this section.

\begin{proof}[Proof of \cref{thm:negwit-capacitance}]
 First, we prove that any unit $st$-potential on $G(x)$ can be
transformed into a negative witness for $x$ in $P_{G}$ with witness size equal
to twice the unit potential energy of that potential. This shows that
$w_-(x,P_{G})\leq 2C_{s,t}(G(x))$.

Given a unit $st$-potential $\sop V:V(G)\rightarrow \mathbb{R}$ on $G(x)$, we consider $\omega_{\sop V}=\sum_{v\in V(G)}{\sop V}(v)\bra{v}$. Then because 
$\sop V(s)=1$ and ${\sop V}(t)=0$, we have $\omega_{\sop V}\tau=1$. Secondly,
\begin{align}
\omega_{\sop V}A\Pi_{H(x)}
&=\sum_{u'\in V(G)}{\sop V}(u')\bra{u'}\sum_{(u,v,\edgeL )\in \overrightarrow{E}(G(x))}\sqrt{c(u,v,\edgeL)}(\ket{u}-\ket{v})\bra{u,v,\edgeL }\nonumber\\
&=\sum_{(u,v,\edgeL )\in \overrightarrow{E}(G(x))}\sqrt{c(u,v,\edgeL)}({\sop V}(u)-{\sop V}(v))\bra{u,v,\edgeL }
=0,
\end{align}
where we've used the definition of unit $st$-potential, which states that $\sop V(u)-\sop V(v)=0$ when $(u,v,\edgeL)\in E(G(x)).$ Thus $\omega_{\sop V}$ is a valid negative witness for input $x$.

We have
\begin{align}
w_-(x,P_{G})\leq\min_{\sop V}\|\omega_{\sop V} A\|^2&=\min_{\sop V}\left\|\sum_{(u,v,\edgeL )\in \overrightarrow{E}(G)}\sqrt{c(u,v,\edgeL)}(\sop V(u)-\sop V(v))\bra{u,v,\edgeL }\right\|\nonumber\\
&=2\min_{\sop V}\sum_{(u,v,\edgeL )\in E(G)}(\sop V(u)-\sop V(v))^2c(u,v,\edgeL)
=2C_{s,t}(G(x)),\label{eqn:conductance_bound}
\end{align}
where the minimization is over unit $st$-potentials on $G(x).$

Next, we show that any negative witness $\omega$ for $P_{G}$ on input $x$  can
be transformed into a unit $st$-potential $\sop V_\omega$ on $G(x)$, with negative
witness size equal to twice the unit potential energy of $\sop V_\omega.$ This shows that
$w_-(x,P_{G})\geq 2C_{s,t}(G(x))$.

Given $\omega$, a negative witness for input $x$, let $\sop V_\omega(v)=\omega(\ket{v}-\ket{t})$ for $v\in
V(G).$ Then $V_\omega(s)=\omega(\ket{s}-\ket{t})=\omega\tau=1$, and
$V_\omega(t)=\omega(\ket{t}-\ket{t})=0$. Also for $(u,v,\edgeL)\in E(G(x))$,
we have
\begin{align}
 V_\omega(u)-V_\omega(v)&=\omega(\ket{u}-\ket{t})-\omega(\ket{v}-\ket{t})=\omega(\ket{u}-\ket{v})\nonumber\\
 &=\omega A \ket{(u,v,\edgeL)}=\omega A \Pi_{H(x)}\ket{(u,v,\edgeL)}=0,
\end{align}
because $\omega A \Pi_{H(x)}=0.$ Thus, $V_\omega$ is a $st$-unit potential for $G(x)$. 

Then 
\begin{align}
w_-(x,P_{G})=&\min_{\omega}\|\omega A\|^2=\min_{\omega}\sum_{(u,v,\edgeL )\in \overrightarrow{E}(G)}(\omega(\ket{u}-\ket{v}))^2c(u,v,\edgeL)\nonumber\\
=&2\min_{\omega}\sum_{(u,v,\edgeL )\in E(G)}(\sop V_\omega(u)-\sop V_\omega(v))^2c(u,v,\edgeL)\geq 2C_{s,t}(G(x)),
\end{align}
where the minimization is over negative witnesses. 
Since $w_-(x,P_{G})\geq 2C_{s,t}(G(x))$ and $w_-(x,P_{G})\leq 2C_{s,t}(G(x))$, we must have $w_-(x,P_{G})= 2C_{s,t}(G(x))$.
\end{proof}


\section{Applications}\label{sec:Applications}
\subsection{Estimating the Capacitance of a Circuit}\label{sec:estimatingCapacitance}
By \cref{thm:negwit-capacitance},  $w_-(x,P_{G})=2C_{s,t}(G(x))$,
so we can apply \cref{theorem:span-est} to estimate $C_{s,t}(G(x))$. By
\cref{theorem:span-est}, the complexity of doing this depends on
$C_{s,t}(G(x))$ and $\widetilde{W}_+(P_{G})=\max_x
\tilde{w}_+(x,P_{G})$. We will prove the following theorem:

\begin{theorem}\label{thm:PosAppWitBound}
For the span program $P_{G}$, we have that $\widetilde{W}_+(P_{G})=O(\max_pJ_{E(G)}(p))$, where the maximum runs over all $st$-unit flows $p$ that are paths from $s$ to $t$. 
\end{theorem}
Note that when the weights are all 1, $\max_p J_E(G)(p)$ is just the length of the longest self-avoiding $st$-path in $G$. 
Combining  \cref{theorem:span-est}, \cref{thm:negwit-capacitance}, and \cref{thm:PosAppWitBound}, we have:
\begin{corollary}\label{cor:estimating-capacitance-query}
Given a network $(G,c)$, with $s,t\in V(G) $ 
 and access to an oracle $O_x$, the bounded error quantum query complexity of estimating $C_{s,t}(G(x))$ to accuracy $\epsilon$ is $\widetilde{O}(\epsilon^{-3/2}\sqrt{C_{s,t}(G(x))\max_pJ_{E(G)}(p)})$ where the maximum runs over all $st$-unit flows $p$ that are paths from $s$ to $t$.
\end{corollary}

\noindent Similarly, we can show:
\begin{corollary}\label{cor:estimating-capacitance-time}
Let $\mathsf{U}$ be the cost of implementing the map 
\begin{equation*}
\ket{u}\ket{0}\mapsto \sum_{v,\ell:(u,v,\ell)\in\overrightarrow{E}(G)}\sqrt{c(u,v,\ell)/d_G(u)}\ket{u,v,\ell}.
\end{equation*}
Then the quantum time complexity of estimating $C_{s,t}(G(x))$ to accuracy $\epsilon$ is $\widetilde{O}(\epsilon^{-3/2}\sqrt{C_{s,t}(G(x))\max_pJ_{E(G)}(p)}\mathsf{U})$. 
\end{corollary}
\begin{proof}
The algorithm in \cref{theorem:span-est} requires $\widetilde{O}(\epsilon^{-3/2}\sqrt{C_{s,t}(G(x))\max_pJ_{E(G)}(p)})$ calls to a unitary $U(P_G,x)$, and other elementary operations \cite{IJ15}. By \cite{JK2017} (generalizing \cite{BR12}), for any $G$, $U(P_G,x)$ can be implemented in cost $\mathsf{U}$. 
\end{proof}

To prove \cref{thm:PosAppWitBound}, we first relate unit $st$-flows on $G$ to approximate positive witnesses. Intuitively, an approximate positive witness is an $st$-flow on $G$ that has energy as small as possible on edges in $E(G)\setminus E(G(x))$. Thus, we can upper bound the approximate positive witness size by the highest possible energy of any $st$-flow on $G$, which is always achieved by a flow that is an $st$-path. 

The following claim can be proven using the technique of the proof of \cite[Lemma 11]{JK2017}:
\begin{claim}
Let $P_{G}$ be the span program of \cref{eq:st-conn-span-program}. Then the positive error of $x$ in $P_{G}$ is 
\begin{equation}\label{eq:minErrorFlor}
e_+(x,P_{G})=\min_\theta \{J_{E(G)\setminus E(G(x))}(\theta)\}
\end{equation}
where $\theta$ runs over unit $st$-flows on $G$. The approximate positive witness size is
\begin{align}\label{eq:optFlow}
\tilde{w}_+(x,P_{G})=\min_\theta J_{E(G)}(\theta)
\end{align}
where $\theta$ runs over unit $st$-flows on $G$ such that $J_{E(G)\setminus E(G(x))}(\theta)={e}_+(x,P_{G})$.
\end{claim}

\begin{proof}[Proof of \cref{thm:PosAppWitBound}]
This theorem follows immediately from the observation that $\min_\theta J_{E}(\theta) \leq J_{E}(\theta')$ for any valid $st$-flow $\theta'$ on the set of edges $E$. Hence, we are free to choose $\theta'$ to be any $st$-path through $E$. Then, separately taking $E \mapsto E(G)$ and $E \mapsto E(G)\setminus E(G(x))$ yields the desired result. 
\end{proof}

\subsection{Deciding Connectivity}\label{sec:connectivity}

Let \textsc{conn}$_{G,X}$ be the problem of deciding, given $x\in X$, whether $G(x)$ is connected. That is:
\begin{align}\label{eq:connProb}
\textsc{conn}_{G,X}=\bigwedge_{\{u,v\}:u,v\in V(G)}uv\textrm{-}\textsc{conn}_{G,X}.
\end{align}

Using the technique of converting logical \textsc{and} into $st$-connectivity
problems in series \cite{Nisan:1995:SLC:225058.225101,JK2017}, we note that the above problem is equivalent
to $n(n-1)/2$ $st$-connectivity problems in series, one for each pair of
distinct vertices in $V(G)$. (The approach in Ref. \cite{Arins2016}  is
similar, but only looks at $n-1$ instances --- the pairs $s$ and $v$ for each
$v\in V(G).$ Our approach is symmetrized over the vertices and thus makes the analysis simpler.)

More precisely, we define a graph $\mathcal{G}$ such that:
\begin{align}
V(\mathcal{G})&=V(G)\times \{\{u,v\}:u\neq v\in V(G)\}\nonumber\\
E(\mathcal{G})&=E(G)\times\{\{u,v\}:u\neq v\in V(G)\}
\end{align}
where $\times$ denotes the Cartesian product, and $\{u,v\}$ is an extra label denoting that that edge or vertex is in the $\{u,v\}\tth$ copy
of the graph $G$ present as a subgraph in $\mathcal{G}$. Choose any labeling
of the vertices from $1$ to $n$ (with slight abuse of notation, we use $u$
both for the original vertex name and the label). We then label the vertex
$(1,\{1,2\})$ as $s$ and the vertex $(n, \{n-1,n\})$ as $t$. Next identify
vertices $(v,\{u,v\})$ and $(u,\{u,v+1\})$ if $u<v$ and $v<n$, and identify
vertices $(v,\{u,v\})$ and $(u+1,\{u+1,u+2\})$ if $v=n$ and $u<n-1$. See \cref{fig:mathcalG} for an example of this construction.

\begin{figure}[ht]
\centering
\begin{tikzpicture}[scale = .95]
\node at (0,0) {\begin{tikzpicture}[scale=.8]

\node at (-4,1.3) {$3$};
\node at (-2.7,0) {$2$};
\node at (-5.4,0) {$1$};
\draw [green] plot [smooth] coordinates{(-3,0) (-4,1)};
\draw [blue] plot [smooth] coordinates{(-3,0) (-4,-.5) (-5,0)};
\draw [red] (-3,0)--(-5,0);
\draw [orange] (-4,1)--(-5,0);
\filldraw (-5,0) circle (.1);
\filldraw (-3,0) circle (.1);
\filldraw (-4,1) circle (.1);

\draw [red] (0,0)--(2,0);
\draw [orange] plot [smooth] coordinates{(0,0) (1,1)};
\draw [green] plot [smooth] coordinates{(1,1) (2,0)};
\draw [blue] plot [smooth] coordinates{(0,0) (1,-.5) (2,0)};
\draw [blue] plot [smooth] coordinates{(2,0) (3,1)};
\draw [green] plot [smooth] coordinates{(3,1) (4,0)};
\draw [green] (4,0)--(6,0);
\draw [orange] (2,0)--(4,0);
\draw [red] plot [smooth] coordinates{(2,0) (2.7,.2) (3,1)};
\draw [red] plot [smooth] coordinates{(4,0) (4.7,.2) (5,1)};
\draw [blue] plot [smooth] coordinates{(4,0) (5,1)};
\draw [orange] plot [smooth] coordinates{(5,1) (6,0)};

\filldraw (0,0) circle (.1);
\filldraw (2,0) circle (.1);
\filldraw (1,1) circle (.1);
\filldraw (4,0) circle (.1);
\filldraw (3,1) circle (.1);
\filldraw (6,0) circle (.1);
\filldraw (5,1) circle (.1);

\node at (-.4,0) {$s$};
\node at (6.4,0) {$t$};

\node at (1,1.8) {$\{1,2\}$};
\node at (3,1.8) {$\{1,3\}$};
\node at (5,1.8) {$\{2,3\}$};

\node at (-4,-1.5) {$G$};
\node at (3,-1.5) {$\mathcal{G}$};
\end{tikzpicture}};
\end{tikzpicture}
\caption{Example of how $\mathcal{G}$ is formed from a graph $G$. We have labeled the subgraphs of $\mathcal{G}$ according to the $st$-connectivity problems the subgraphs represent.}\label{fig:mathcalG}
\end{figure}
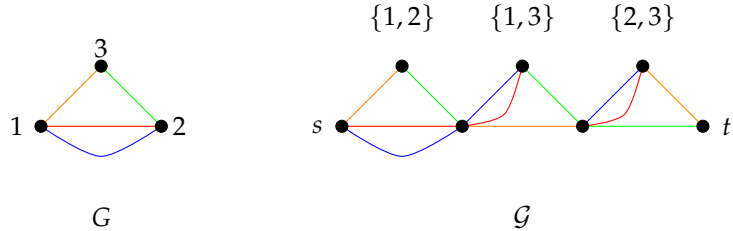

Finally, we define $\mathcal{G}(x)$ to be the subgraph of $\mathcal{G}$ with edges
\begin{align}
E(\mathcal{G}(x))&=E(G(x))\times\{\{u,v\}:u\neq v\in V(G)\}.
\end{align}
We can see that any $st$-path in ${\cal G}(x)$ must go through each of the copies of $G(x)$, meaning it must include, for each $\{u,v\}$, a $uv$-path through the copy of $G(x)$ labeled $\{u,v\}$. Thus, there is an $st$-path in ${\cal G}(x)$ if and only if $G(x)$ is connected.

We consider the span program $P_{\mathcal{G}}$, where $c(e)=1$ for all
$e\in E(\mathcal{G})$. We will use $P_{\cal G}$ to solve $st$-connectivity on ${\cal G}(x)$. To analyze the resulting algorithm, we need to upper bound the negative and positive witness sizes $w_-(x,P_{\cal G})=2C_{s,t}({\cal G}(x))$ and $w_+(x,P_{\cal G})=\frac{1}{2}R_{s,t}({\cal G}(x))$.

\begin{lemma}\label{lem:kirchhoff}
For any $x$ such that $G(x)$ is connected, $w_+(x,P_{\cal G}) = \frac{n(n-1)}{2}R_{\mathrm{avg}}(G(x))$, where $R_{\mathrm{avg}}(G(x))$ is the average resistance.
\end{lemma}
\begin{proof}
Using the rule that resistances in series add, we have: 
\begin{align}
R_{s,t}({\cal G}(x))=\frac{1}{2}\sum_{u,v\in V(G)}R_{u,v}(G(x)) = n(n-1)R_{\mathrm{avg}}(G(x)).\label{eq:kirchhoff}
\end{align}
This is equal to $2w_+(x,P_{\cal G})$. 
\end{proof}

Now we bound $C_{s,t}(\mathcal{G}(x))$, to prove the following: 
\begin{lemma}\label{lem:negative-witness}
Fix $\kappa >1$, and suppose $G(x)$ has $\kappa$ connected components. Then if $G$ is a subgraph of a complete graph (that is, $G$ has at most one edge between any pair of vertices), we have $w_-(x,P_{\cal G})=O(1/\kappa)$. Otherwise, we have $w_-(x,P_{\cal G})=O(d_{\max}(G)/\sqrt{n\kappa})$.
\end{lemma}
\begin{proof}

Using the rule for capacitors in series, and accounting for double counting pairs of vertices, we have
\begin{align}\label{eq:seriesCapSum}
\frac{1}{C_{s,t}(\mathcal{G}(x))}=\frac{1}{2}\sum_{s',t'\in V(G)}\frac{1}{C_{s',t'}(G(x))}.
\end{align}

To put an upper bound on $C_{s,t}(\mathcal{G}(x))$, we can put an upper bound on each term $C_{s',t'}(G(x))$. To upper bound $C_{s',t'}(G(x))$, consider the following unit $s't'$-potential for $G(x)$. Set $\sop V(v)=1$ for all $v$ in the same connected component as $s'$ in $G(x)$, set $\sop V(v)=0$ for all $v$ in the same connected component as $t'$ in $G(x)$, and set ${\sop V}(v)= \nu$ for all other vertices in $G$.
We now find the minimum unit potential energy of this $s't'$-potential, (minimizing over $\nu$). This will be an upper bound on $C_{s',t'}(G(x))$ by \cref{def:effCap}, since it is not necessarily the optimal choice to set all vertices not connected to $s'$ or $t'$ to have the same unit potential value.

While we could calculate the unit potential energy using \cref{def:unitPotentEnergy}, and then minimize that quantity over $\nu,$ we instead use \cref{claim:parallel_series}, and the fact that our choice of unit $s't'$-potential effectively creates a graph with three vertices: one
vertex corresponds to the connected component of $G(x)$ containing $s'$ (let
$n_{s'}$ be the number of vertices in this component), one vertex corresponds
to the connected component of $G(x)$ containing $t'$ (let $n_{t'}$ be the
number of vertices in this component), and one vertex corresponds to all the
other vertices in the graph (let $n_{a}=n-n_{s'}-n_{t'}$ be the remaining
number of vertices).

For any pair of vertices $u$ and $v$, let $D_u$ be the number of edges of $G$
coming out of the component of $G(x)$ containing $u$, and let $D_{uv}$ be the
number of edges in $G$ between the components of $G(x)$ containing $u$ and
$v$, so the number of edges between the component containing $a$ and all components other than that containing $s'$ is $D_{s'}-D_{s'a}$. Using the rules for calculating
capacitance in series and parallel, we have
\begin{align}
C_{s',t'}(G(x))\leq D_{s't'}+\left(\frac{1}{D_{s'}-D_{s't'}}+\frac{1}{D_{t'}-D_{s't'}}\right)^{-1}=\frac{D_{s'}D_{t'}-D_{s't'}^2}{D_{s'}+D_{t'}-2D_{s't'}}.
\end{align}
Using \cref{eq:seriesCapSum}, we have 
\begin{align}\label{eq:CapSumVerts}
\frac{1}{C_{s,t}(\mathcal{G}(x))}\geq
\frac{1}{2}\sum_{\substack{s',t'\in V(G) \\ \{s',t'\} \notin E(G(x))}}\frac{D_{s'}+D_{t'}-2D_{s't'}}{D_{s'}D_{t'}-D_{s't'}^2}.
\end{align}

Now the expression on the right-hand side of \cref{eq:CapSumVerts} depends
only on which connected components $s'$ and $t'$ are in, so instead of summing
over the vertices of $G$, we can instead sum over the $\kappa$ connected
components of $G(x)$. Let $n_i$ be the number of vertices in the $i\tth$
connected component. Then, continuing from \cref{eq:CapSumVerts} we have:
\begin{align}
\frac{1}{C_{s,t}({\cal G}(x))} &\geq \frac{1}{2}\sum_{i,j\in [\kappa]:i\neq j}n_in_j\frac{D_i+D_j - 2D_{ij}}{D_iD_j - D_{ij}^2} = \frac{1}{2}\sum_{i,j\in [\kappa]:i\neq j}n_in_j\frac{(D_i-D_{ij})+(D_j - D_{ij})}{D_iD_j - D_{ij}^2}\nonumber
\\&=\frac{1}{2}\left(2\sum_{i,j\in [\kappa]:i\neq j}n_in_j\frac{(D_i-D_{ij})}{D_iD_j - D_{ij}^2}\right) = \sum_{i,j\in[\kappa]:i\neq j}n_in_j\frac{D_j-D_{ij}}{D_iD_j-D_{ij}^2}.
\label{eq:branch}
\end{align}
First, consider the case when $G$ is a complete graph. In that case, $D_i-
D_{ij}$, the number of edges leaving component $i$ and going to a component
other than $i$ or $j$, is exactly $n_i(n-n_i-n_j)$, whereas the number of
edges leaving component $i$ is $D_i= n_i(n-n_i)$. Finally, $D_{ij}=n_in_j$ is
the number of edges going from the $i\tth$ component to the $j\tth$ component.
Thus, continuing from \cref{eq:branch}, we have:
\begin{eqnarray}
\frac{1}{C_{s,t}({\cal G}(x))}
&\geq & \sum_{i,j\in[\kappa]:i\neq j}n_in_j\frac{n_i(n-n_i-n_j)}{n_i(n-n_i)n_j(n-n_j)-n_i^2n_j^2}
= \sum_{i,j\in[\kappa]:i\neq j}\frac{n_i(n-n_i-n_j)}{(n-n_i)(n-n_j)-n_in_j}\nonumber\\
&=& \sum_{i,j\in[\kappa]:i\neq j}\frac{n_i(n-n_i-n_j)}{n^2-nn_i-nn_j}
= \sum_{i,j\in[\kappa]:i\neq j}\frac{n_i}{n} = \kappa-1.
\end{eqnarray}
Note that this upper bound on $C_{s,t}({\cal G}(x))$ applies to \emph{any}
subgraph of a complete graph, since adding edges can only increase the
capacitance. Thus, we have completed the first part of the proof.

We now continue with the more general case, where $G$ is not necessarily a
subgraph of a complete graph. Let $d=d_{\max}(G)$. Continuing from
\cref{eq:branch}, and using the fact that for any component, we have:
\begin{eqnarray}
\frac{1}{C_{s,t}({\cal G}(x))}&\geq &
\frac{1}{2}\sum_{i,j\in [\kappa]:i\neq j}n_in_j\frac{D_i+D_j - 2D_{ij}}{D_iD_j - D_{ij}^2}
= \frac{1}{2}\sum_{i,j\in [\kappa]:i\neq j}n_in_j\frac{(\sqrt{D_i}-\sqrt{D_j})^2 + 2\sqrt{D_iD_j}-2D_{ij}}{D_iD_j - D_{ij}^2}\nonumber\\
&\geq& \sum_{i,j\in [\kappa]:i\neq j}n_in_j\frac{\sqrt{D_iD_j}-D_{ij}}{{D_iD_j} - D_{ij}^2}
\geq \sum_{i,j\in [\kappa]:i\neq j}n_in_j\frac{\sqrt{D_iD_j-D_{ij}^2}}{{D_iD_j} - D_{ij}^2}
\geq \sum_{i,j\in [\kappa]:i\neq j}n_in_j\frac{1}{\sqrt{D_iD_j}}\nonumber\\
&\geq & \sum_{i,j\in[\kappa]:i\neq j}\frac{\sqrt{n_in_j}}{d}
\geq \frac{1}{d}\sqrt{\sum_{i\in [\kappa]}\sum_{j\neq i}n_in_j}
=\frac{1}{d}\sqrt{\sum_{i\in [\kappa]}n_i(n-n_i)}.
\end{eqnarray}
Above we used the fact that for any component, $D_i\leq dn_i$. 
The sum $\sum_{i\in[\kappa]}n_i^2$ is maximized when the $n_i$ are as far as possible from uniform. In this case, we have $n_i\geq 1$ for all $i$, so $\sum_{i\in[\kappa]}n_i^2\leq (\kappa-1)+(n-(\kappa-1))^2$. Thus, continuing, we have
\begin{eqnarray}
\frac{1}{C_{s,t}({\cal G}(x))}&\geq& \frac{1}{d}\sqrt{n\sum_{i\in [\kappa]}n_i-\sum_{i\in [\kappa]}n_i^2}
\geq \frac{1}{d}\sqrt{n^2-(\kappa-1) - n^2 - (\kappa-1)^2 + 2n(\kappa-1)}\nonumber\\
&=& \frac{1}{d}\sqrt{(2n-\kappa)(\kappa-1)} 
\geq \frac{1}{d}\sqrt{n\kappa}.
\end{eqnarray}
The result follows by \cref{thm:negwit-capacitance}, which says that $w_-(x,P_{\cal G})=2C_{s,t}({\cal G}(x))$.
\end{proof}

Combining \cref{lem:negative-witness,lem:kirchhoff,thm:span-decision}, we have the following:
\begin{theorem}\label{thm:connectivity-algorithm}
For any family of graphs $G$ such that $G$ is a subgraph of a complete graph,
and $X\subseteq\{0,1\}^{E(G)}$ such that for all $x\in X$, if $G(x)$ is
connected, $R_{\mathrm{avg}}(G(x))\leq R$, and if $G(x)$ is not connected, it
has at least $\kappa$ components, the bounded error quantum query complexity
of $\textsc{conn}_{G,X}$ is $O\left(n\sqrt{R/\kappa}\right)$.

For any family of connected graphs $G$ and $X\subseteq\{0,1\}^{E(G)}$ such
that for all $x\in X$, if $G(x)$ is connected, $R_{\mathrm{avg}}(G(x))\leq R$,
and if $G(x)$ is not connected, it has at least $\kappa$ components, the
bounded error quantum query complexity of $\textsc{conn}_{G,X}$ is
$O\left(n^{3/4}\sqrt{Rd_{\max}(G)}/\kappa^{1/4}\right)$.
\end{theorem}

Similarly, we can show:
\begin{corollary}\label{cor:connectivity-algorithm}
Let $\mathsf{U}$ be the cost of implementing the map 
\begin{equation*}
\ket{u}\ket{0}\mapsto \sum_{v,\ell:(u,v,\ell)\in\overrightarrow{E}(G)}\sqrt{1/d_G(u)}\ket{u,v,\ell}.
\end{equation*}
If $G$ is subset of a complete graph, the quantum time complexity of $\textsc{conn}_{G,X}$ is $O(n\sqrt{R/\kappa}\mathsf{U})$.

For any family of connected graphs $G$ and $X\subseteq\{0,1\}^{E(G)}$ such
that for all $x\in X$, if $G(x)$ is connected, $R_{\mathrm{avg}}(G(x))\leq R$,
and if $G(x)$ is not connected, it has at least $\kappa$ components, the
quantum time complexity of $\textsc{conn}_{G,X}$ is
$O\left(n^{3/4}\sqrt{Rd_{\max}(G)}/\kappa^{1/4}\mathsf{U}\right)$.
\end{corollary}
\begin{proof}
The algorithm of \cref{thm:span-decision} makes $O\left(n^{3/4}\sqrt{Rd_{\max}(G)}/\kappa^{1/4}\right)$ calls to a unitary $U(P_{\cal G},x)$ (see \cite{IJ15}). 
By \cite{JK2017} (generalizing \cite{BR12}), for any $G$, $U(P_{\cal G},x)$ can be implemented in cost $\mathsf{U}$. 
\end{proof}


\section{Spectral Algorithm for Deciding Connectivity}\label{SpecConnectivity}

In this section, we will give alternative quantum algorithms for deciding
connectivity. We begin by presenting an algorithmic template, outlined in \cref{alg:alg1}, that requires the instantiation of a certain
initial state. Since this initial state is independent of the input, we
already get an upper bound on the quantum query complexity, as follows:

\begin{restatable}{corollary}{query}\label{cor:query}
Fix any $\lambda>0$ and $\kappa >1$. For any family of connected graphs $G$ on $n$ vertices and $X\subseteq \{0,1\}^{E(G)}$ such that for all $x\in X$, either $\lambda_2(G(x))\geq \lambda$ or $G(x)$ has at least $\kappa$ connected components, the bounded error quantum query complexity of $\textsc{conn}_{G,X}$ is $O\left(\sqrt{\frac{n d_{\max}(G)}{\kappa\lambda}}\right)$.
\end{restatable}

In \cref{sec:any-G}, we describe one such initial state, and how to prepare it, leading to the following upper bound, in which $\mathsf{U}$ is the cost of performing one step of a quantum walk on $G$, and $\mathsf{S}$ is the cost of preparing a quantum state corresponding to the stationary distribution of a quantum walk on $G$:

\begin{restatable}{theorem}{anyG}\label{thm:any-G}
Fix any $\lambda>0$ and $\kappa>1$. For any family of connected graphs $G$ on $n$ vertices and
$X\subseteq \{0,1\}^{E(G)}$ such that for all $x\in X$, either
$\lambda_2(G(x))\geq \lambda$, or $G(x)$ has at least $\kappa$ connected
components, $\textsc{conn}_{G,X}$ can be solved in bounded error in time
\begin{equation*}
\widetilde{O}\left(\sqrt{\frac{nd_{\mathrm{avg}}(G)}{\kappa\lambda_2(G)}}\left(\mathsf{S}+\sqrt{\frac{d_{\max}(G)}{\lambda}}\mathsf{U}\right)\right).
\end{equation*}
\end{restatable}

In \cref{sec:Cayley-graph}, we restrict our attention to the case when $G$ is
a Cayley graph, and give an alternative instantiation of the algorithm in \cref{alg:alg1}, proving the following, where $\mathsf{\Lambda}$ is the cost of computing the eigenvalues of $G$:
\begin{restatable}{theorem}{Cayleygraph}\label{thm:Cayley-graph}
Fix any $\lambda>0$ and $\kappa >1$. For any family of connected
graphs $G$ on $n$ vertices such that each $G$ is a degree-$d$ Cayley graph over an Abelian group, and
$X\subseteq \{0,1\}^{E(G)}$ such that for all $x\in X$, either
$\lambda_2(G(x))\geq \lambda$ or $G(x)$ has at least $\kappa$ connected components,
$\textsc{conn}_{G,X}$ can be solved in bounded error in time $$\widetilde{O}
\left(\sqrt{\frac{nd}{\kappa\lambda}}\mathsf{U}+\sqrt{\frac{nd}{\kappa\lambda_2(
G)}}\mathsf{\Lambda}\right).$$
\end{restatable}

We remark that the results in this section, in contrast to the previous
connectivity algorithm, apply with respect to any weighting of the edges of
$G$. Applying non-zero weights to the edges of $G$ does not change which
subgraphs $G(x)$ are connected, but it does impact the complexity of our
algorithm. Thus, for any weight function on the edges, we get algorithms with
the complexities given in \cref{cor:query}, \cref{thm:any-G} and 
\cref{thm:Cayley-graph}, where $d_{\max}(G)$ and $d_{\mathrm{avg}}(G)$ are in terms of
the weighted degrees, and $\lambda_2(G)$ and $\lambda_2(G(x))$ are in terms of
the weighted Laplacians.

Finally, in \cref{sec:estimating-connectivity}, we describe how when $G$ is a complete graph, these ideas
can be used to design algorithms, not only for deciding connectivity, but also for
estimating the \emph{algebraic connectivity} of a graph, which is a measure of
how connected a graph is. In particular, we show:
\begin{restatable}{theorem}{ConnEstimation}\label{thm:Connectivity-estimation}
Let $G$ be the complete graph on $n$ vertices. There
exists a quantum algorithm that, on input $x$, with probability at least 2/3, outputs an estimate $\tilde\lambda$ such that $\abs{\tilde\lambda-\lambda_2(G(x))}\leq \varepsilon \lambda_2(G(x))$, where $\lambda_2(G(x))$ is 
the $\textit{algebraic connectivity}$ of $G(x)$, in time
$\widetilde{O}\left(\frac{1}{\varepsilon}\frac{n}{\sqrt{\lambda_2(G(x))}}\right)$.
\end{restatable}

\vskip10pt
Let $P_{G}=(H,U,A,\tau)$ be the span program for $st$-connectivity defined in
\cref{eq:st-conn-span-program}. Note that only $\tau$ depends on $s$ and $t$,
and we will not be interested in $\tau$ here. We let $A(x)=A\Pi_{H(x)}$.
\begin{eqnarray}
A(x)A(x)^T &=& \sum_{(u,v,\ell)\in \overrightarrow{E}(G(x))}\sqrt{c(u,v,\ell)}(\ket{u}-\ket{v})\bra{u,v,\ell} \sum_{(u,v,\ell)\in \overrightarrow{E}(G(x))}\sqrt{c(u,v,\ell)}\ket{u,v,\ell}(\bra{u}-\bra{v})\nonumber\\
&=& \sum_{(u,v,\ell)\in\overrightarrow{E}(G(x))}c(u,v,\ell)(\ket{u}\bra{u}-\ket{u}\bra{v}-\ket{v}\bra{u}+\ket{v}\bra{v})\nonumber\\
&=& \sum_{u\in [n]}2d_G(u)\ket{u}\bra{u} - 2 {\cal A}_{G(x)} = 2({\cal D}_{G(x)}-{\cal A}_{G(x)}) = 2L_{G(x)}.\label{eq:AL}
\end{eqnarray}
Above, $L_{G(x)}$ is the Laplacian of $G(x)$ (see \cref{sec:graphPrelim}). By
a similar computation to the one above, we have $AA^T=2L_G$, where $G$ is the
parent graph, upon which $A$ depends. Recall that for any $G$, the eigenvalues
of $L_G$ lie in $[0,d_{\max}]$, with
$\ket{\mu}=\frac{1}{\sqrt{n}}\sum_v\ket{v}$ as a 0-eigenvalue. In our case,
since $G$ is assumed to be connected, $\ket{\mu}$ is the only 0-eigenvector of
$L_G$, so $\mathrm{row}(L_G)$ is the orthogonal complement of $\ket{\mu}$. For
any $x$, $G(x)$ also has $\ket{\mu}$ as a 0-eigenvalue, but if $G(x)$ is
connected, this is the only 0-eigenvalue. In general, the dimension of the
0-eigenspace of $L_{G(x)}$ is the number of components of $G(x)$. Thus,
Eq.~\eqref{eq:AL} implies the following.
\begin{itemize}
\item The multiset of nonzero eigenvalues of $L_G$ are exactly half of the squared singular values of $A$, and in particular, since no eigenvalue of $L_G$ can be larger than the maximum degree of $G$, $\sigma_{\max}(A)\leq \sqrt{2d_{\max}(G)}$.
\item The multiset of nonzero eigenvalues of $L_{G(x)}$ are exactly half the
squared singular values of $A(x)$, and in particular, if $G(x)$ is connected,
then $\sigma_{\min}(A(x))=\sqrt{2\lambda_2(G(x))}$, where $\lambda_2(G(x))$ is
the second smallest eigenvalue of $L_{G(x)}$, which is non-zero if and only if
$G(x)$ is connected.
\item The support of $L_G$ is $\mathrm{col}(A)$, which is the orthogonal
subspace of the uniform vector $\ket{\mu}=\frac{1}{\sqrt{n}}\sum_v\ket{v}$.
\end{itemize}

For a particular span program $P$, and input $x$, an associated unitary
$U(P,x)=(2\Pi_{\ker A}-I)(2\Pi_{H(x)}-I)$ can be used to construct quantum
algorithms, for example, for deciding the span program. Then by \cite[Theorem
3.10]{IJ15}, which states that $\Delta(U(P,x))\geq
2\sigma_{\min}(A(x))/\sigma_{\max}(A)$, we have the following.

\begin{lemma}\label{lem:Delta-bound}
Let $P_G$ be the $st$-connectivity span program from \cref{eq:st-conn-span-program}. Then $\Delta(U(P,x))\geq 2\sqrt{\lambda_2(G(x))/d_{\max}(G)}$.
\end{lemma}

Our algorithm will be based on the following connection between the
connectivity of $G(x)$ and the presence of a 0-phase eigenvector of $U(P,x)$
in $\mathrm{row}(A)$.
\begin{lemma}\label{lem:kappa}
$G(x)$ is not connected if and only if there exists $\ket{\psi}\in\mathrm{row}(A)$ that is fixed by $U(P,x)$. Moreover, if $G(x)$ has $\kappa>1$ components, there exists a $(\kappa-1)$-dimensional subspace of $\mathrm{row}(A)$ that is fixed by $U(P,x)$. 
\end{lemma}
\begin{proof}

If $G(x)$ is not connected, and in particular, $G(x)$ has $\kappa >1$ components, then 
\begin{equation}\label{rankA}
\mathrm{rank}(A(x))=\mathrm{rank}(A(x)A(x)^T)=\mathrm{rank}(L_{G(x)})=n-\kappa < n-1 = \mathrm{rank}(A).
\end{equation}
Eq.~\eqref{rankA} implies $\ker A\subseteq \ker A(x)$, which means $\ker A(x)\cap \mathrm{row}(A)$
has dimension $\kappa-1 \geq 1$. Let $\ket{\psi}\in \ker A(x)\cap
\mathrm{row}(A)$. Since $\ket{\psi}\in \ker A(x)$, $A\Pi_{H(x)}\ket{\psi}=0$,
so $\Pi_{H(x)}\ket{\psi}\in \ker A$. Since $\ket{\psi}\in\mathrm{row}(A)$, 
$\ket{\psi}=\Pi_{\mathrm{row}(A)}\ket{\psi}=\Pi_{\mathrm{row}(A)}\Pi_{H(x)^\bot}\ket{\psi}$, so $\ket{\psi}\in H(x)^\bot$. Thus since $\ket{\psi}\in
H(x)^\bot\cap\mathrm{row}(A)$, it follows that $U(P,x)\ket{\psi}=\ket{\psi}$.

On the other hand, if $G(x)$ is connected, then both $A(x)$ and $A$ have rank
$n-1$, so $\ker A=\ker A(x)$ because $\ker A\subseteq \ker A(x)$. If
$\ket{\psi}\in\mathrm{row}(A)$ is fixed by $U(P,x)$, then $\ket{\psi}\in
H(x)^\bot$, meaning $A(x)\ket{\psi}=0$, hence $\ket{\psi}\in \ker A(x)=\ker A$,
which is a contradiction.
\end{proof}

Thus, to determine if $G(x)$ is connected, it is sufficient to detect the presence of \emph{any} 0-phase eigenvector of $U(P,x)$ on $\mathrm{row}(A)$. 
Let $\{\ket{\psi_i}\}_{i=1}^{n-1}$ be any basis for $\mathrm{row}(A)$, not
necessarily orthogonal, and suppose we have access to an operation that
generates
$$\ket{\psi_{\mathrm{init}}}=\sum_{i=1}^{{n-1}}\ket{i}\ket{\psi_i}.$$ Such a
basis is independent of the input, so we can certainly perform such a map with
0 queries. We will later discuss cases in which we can implement such a map
time efficiently.

\begin{algorithm}\label{alg:alg1}
Assume there is a known constant $\lambda$ such that if $G(x)$ is connected, then $\lambda_2(G(x))\geq \lambda$. Let $\{\ket{\psi_i}\}_i$ be some states that span the rowspace of $A$, whose choice determines the cost of the amplitude estimation step.
\begin{enumerate}
\item Prepare $\ket{\psi_{\mathrm{init}}}=\sum_{i=1}^{n-1}\frac{1}{\sqrt{n-1}}\ket{i}\ket{\psi_i}$.
\item Perform the phase estimation of $U(P,x)$ (see \cref{thm:phase-estimation}) on the second register, to precision $\sqrt{\lambda/d_{\max}(G)}$, and accuracy $\epsilon$. 
\item Use amplitude estimation (see \cref{thm:amplitude-estimation}) to determine if the amplitude on $\ket{0}$ in the phase register is $0$, in which case, output ``connected'', or $>0$, in which case, output ``not connected.''
\end{enumerate}
\end{algorithm}

The algorithm proceeds by first
preparing the initial state $\ket{\psi_{\mathrm{init}}}=\sum_{i=1}^{n-1}\frac{
1}{\sqrt{n-1}}\ket{i}\ket{\psi_i}$. Next, the algorithm performs phase
estimation on the second register, as described in 
\cref{thm:phase-estimation}, with precision $\sqrt{\lambda/d_{\max}(G)}$ and accuracy
$\epsilon$. First, suppose that there are $\kappa-1>0$ orthonormal 0-phase
eigenvectors of $U(P,x)$ in $\mathrm{row}(A)$, and let $\Pi$ be the
orthonormal projector onto their span. By \cref{thm:phase-estimation}, for
each $i$, the phase estimation step will map
$\ket{i}\left(\Pi\ket{\psi_i}\right)$ to
$\ket{i}\ket{0}\left(\Pi\ket{\psi_i}\right)$. Thus, the squared amplitude on
$\ket{0}$ in the phase register will be at least:
\begin{equation}
\varepsilon:=\frac{\norm{(I\otimes \Pi)\ket{\psi_{\mathrm{init}}}}^2}{\norm{\ket{\psi_{\mathrm{init}}}}^2} = \frac{1}{\norm{\ket{\psi_{\mathrm{init}}}}^2}\sum_{i=1}^{n-1}\norm{\Pi\ket{\psi_i}}^2>0,\label{eq:eps}
\end{equation}
since the $\ket{\psi_i}$ span $\mathrm{row}(A)$. 

On the other hand, suppose $G(x)$ is connected, so there is no 0-phase
eigenvector in $\mathrm{row}(A)$. Then all phases will be at least
$\Delta(U(P,x))\geq \sqrt{\lambda/d_{\max}(G)}$, by \cref{lem:Delta-bound}, so the phase register will
have squared overlap at most $\epsilon$ with $\ket{0}$.

Setting $\epsilon=\varepsilon/2$, we just need to distinguish between an
amplitude of $\geq \varepsilon$ and an amplitude of $\leq \varepsilon/2$ on
$\ket{0}$. Using \cref{cor:amplitude-estimation}, we can distinguish these two
cases in $\frac{1}{\sqrt{\varepsilon}}$ calls to steps 1 and 2. By 
\cref{thm:phase-estimation}, Step 2 can be implemented using
$\sqrt{\frac{d_{\max}(G)}{\lambda}}\log\frac{1}{\varepsilon}$ calls to
$U(P,x)$. By \cite[Theorem 13]{JK2017}, if $\mathsf{U}$ is the cost of
implementing, for any $u\in V$, the map
\begin{equation}
\ket{u,0}\mapsto \sum_{v,\ell: (u,v,\ell)\in\overrightarrow{E}(G)}\sqrt{c(u,v,\ell)/d_G(u)}\ket{u,v,\ell},\label{eq:update}
\end{equation}
which corresponds to one step of a quantum walk on $G$, then $U(P,x)$ can be implemented in time $O(\mathsf{U})$. We thus get the following:
\begin{theorem}\label{thm:alg-template}
Fix $\lambda>0$. Let $\mathsf{Init}$ denote the cost of generating the initial state $\ket{\psi_{\mathrm{init}}}$, and $\mathsf{U}$ the cost of the quantum walk step in \cref{eq:update}. Let $\varepsilon$ be as in \cref{eq:eps}. Then for any family of connected graphs $G$ and $X\subseteq \{0,1\}^{E(G)}$ such that for all $x\in X$, either $\lambda_2(G(x))\geq \lambda$ or $G(x)$ is not connected, $\textsc{conn}_{G,X}$ can be decided by a quantum algorithm with cost $O\left(\frac{1}{\sqrt{\varepsilon}}\left(\mathsf{Init}+\sqrt{\frac{d_{\max}(G)}{\lambda}}\mathsf{U}\log \frac{1}{\varepsilon}\right)\right)$.
\end{theorem}

In \cref{sec:any-G} and \cref{sec:Cayley-graph}, we will discuss particular implementations of this algorithm, but if we only care about query complexity, we already have the following.

\query*
\begin{proof}
First, observe that $U(P,x)$ can be implemented with 2 queries. 

Next, let $\{\ket{\psi_i}\}_{i=1}^{n-1}$ be any orthonormal basis for $\mathrm{row}(A)$. Then in $\mathsf{Init}=0$ queries, we can generate the state 
$$\ket{\psi_{\mathrm{init}}}=\frac{1}{\sqrt{n-1}}\sum_{i=1}^{n-1}\ket{i}\ket{\psi_i}.$$
Then if there are $\kappa-1$ orthonormal 0-phase vectors of $U(P,x)$ in $\mathrm{row}(A)$, $\ket{\phi_1},\dots,\ket{\phi_{\kappa-1}}$, setting $\Pi=\sum_{j=1}^{\kappa-1}\ket{\phi_j}\bra{\phi_j}$, we have 
$$\varepsilon = \norm{(I\otimes \Pi)\ket{\psi_{\mathrm{init}}}}^2 = \frac{1}{n-1}\sum_{j=1}^{\kappa-1}\sum_{i=1}^{n-1}|\braket{\phi_j}{\psi_i}|^2=\frac{\kappa-1}{n-1}.$$
Then the result follows from \cref{thm:alg-template}.
\end{proof}

\subsection{A Connectivity Algorithm for any $G$}\label{sec:any-G}

Let $\ket{\mu}=\frac{1}{\sqrt{n}}\sum_{u\in V(G)}\ket{u}$ and for
$j\in\{1,\dots,n-1\}$ let $\ket{\hat{j}}=\frac{1}{\sqrt{n}}\sum_{u\in
[n]}e^{2\pi iju/n}\ket{u}$. Then it is easily checked that
$\ket{\hat{1}},\dots,\ket{\widehat{n-1}}$ are an orthonormal basis for the
columnspace of $A$ for a connected graph $G$. Let $\ket{\psi_j} =
A^T\ket{\hat{j}}$. Then $\{\ket{\psi_1},\dots,\ket{\psi_{n-1}}\}$ is a basis
for $\mathrm{row}(A)$, but it is not necessarily orthogonal unless the
$\ket{\hat{j}}$ form an eigenbasis of $L_G$, (as in the case of Cayley graphs,
discussed in \cref{sec:Cayley-graph}), and in general, the $\ket{\psi_j}$ are
not normalized. We have
\begin{eqnarray}
\ket{\psi_j}=A^T\ket{\hat{j}}&=&\sum_{(u,v,\ell)\in\overrightarrow{E}(G)}\sqrt{c(u,v,\ell)}\ket{u,v,\ell}(\braket{u}{\hat{j}} - \braket{v}{\hat{j}})\nonumber\\
& =& \sum_{(u,v,\ell)\in\overrightarrow{E}(G)}\sqrt{c(u,v,\ell)}\left(\frac{1}{\sqrt{n}}e^{2\pi i ju/n} - \frac{1}{\sqrt{n}}e^{2\pi i jv/n}\right)\ket{u,v,\ell},
\label{eq:psi-j}
\end{eqnarray}
from which we can compute
\begin{eqnarray}
\norm{\ket{\psi_j}}^2 = \norm{A^T\ket{\hat{j}}}^2 &=& \frac{1}{n}\sum_{(u,v,\ell)\in\overrightarrow{E}(G)}c(u,v,\ell)\left|e^{2\pi iju/n} - e^{2\pi ijv/n} \right|^2\nonumber\\
&=& \frac{1}{n}\sum_{(u,v,\ell)\in\overrightarrow{E}(G)}c(u,v,\ell)2\left(1-\cos\frac{2\pi j(v-u)}{n}\right).\label{eq:norm-psi-j}
\end{eqnarray}

For any graph, we can use as initial state a normalization of $\sum_{j=1}^{n-1}\ket{j}\ket{\psi_j}$. From Eq.~\eqref{eq:norm-psi-j}, we have, using Lagrange's identity:
\begin{eqnarray}
\sum_{j=1}^{n-1}\norm{\ket{\psi_j}}^2 & = & \frac{2}{n}\sum_{(u,v,\ell)\in\overrightarrow{E}(G)}c(u,v,\ell)\sum_{j=1}^{n-1}\left(1-\cos\frac{2\pi j(v-u)}{n}\right)\nonumber\\
&=& \frac{2}{n}\sum_{(u,v,\ell)\in\overrightarrow{E}(G)}c(u,v,\ell)\left(n - 1 - \left(\frac{1}{2}\frac{\sin\left((n-1/2)\frac{2\pi(v-u)}{n}\right)}{\sin\left(\frac{\pi(v-u)}{n}\right)}-1\right)\right)\nonumber\\
&=& \frac{2}{n}\sum_{(u,v,\ell)\in\overrightarrow{E}(G)}c(u,v,\ell)\left(n - \frac{\sin\left(-\frac{\pi(v-u)}{n}\right)}{2\sin\left(\frac{\pi(v-u)}{n}\right)}\right)\nonumber\\
&=& \frac{2}{n}\left(n+\frac{1}{2}\right)\sum_{u\in V}\sum_{v,\ell:(u,v,\ell)\in\overrightarrow{E}(G)}c(u,v,\ell)\nonumber\\
& =& 2\left(1+\frac{1}{2n}\right)\sum_{u\in V}d_G(u)=2\left(1+\frac{1}{2n}\right)nd_{\mathrm{avg}}(G),\label{eq:psi-init-norm}
\end{eqnarray}
where $d_{\mathrm{avg}}=d_{\mathrm{avg}}(G)$ is the average weighted degree in $G$. 

Define the initial state as the unit vector:
\begin{equation}
\ket{\psi_{\mathrm{init}}}=\frac{\sum_{j=1}^{n-1}\ket{j}\ket{\psi_j}}{\sqrt{2(1+\frac{1}{2n})nd_{\mathrm{avg}}}}.\label{eq:psi-init}
\end{equation}
Then we can lower bound the overlap with the 0-phase space of $U(P,x)$ in the case where $G(x)$ is not connected, as follows:
\begin{lemma}\label{lem:eps}
Let $\ket{\psi_{\mathrm{init}}}$ be as in \cref{eq:psi-init}, and suppose $G(x)$ has $\kappa > 1$ connected components. Let $\Pi$ be the projector onto a $(\kappa-1)$-dimensional subspace of $\mathrm{row}(A)$ that is in the 0-phase space of $U(P,x)$. 
Then 
$$\norm{(I\otimes \Pi)\ket{\psi_{\mathrm{init}}}}^2\geq \frac{(\kappa-1)\lambda_2(G)}{(1+\frac{1}{2n})nd_{\mathrm{avg}}}.$$
\end{lemma}
\begin{proof}
By \cref{lem:kappa}, there exist orthonormal vectors $\ket{\phi_1},\dots,\ket{\phi_{\kappa-1}}\in \mathrm{row}(A)$ that are fixed by $U(P,x)$. Let $\Pi=\sum_{i=1}^{\kappa-1}\ket{\phi_i}\bra{\phi_i}$. 
We have:
\begin{equation}
\norm{(I\otimes \Pi)\ket{\psi_{\mathrm{init}}}}^2 
= \frac{\sum_{i=1}^{\kappa-1}\sum_{j=1}^{n-1}|\braket{\phi_i}{\psi_j}|^2}{2(1+\frac{1}{2n})nd_{\mathrm{avg}}}
= \frac{\sum_{i=1}^{\kappa-1}\sum_{j=1}^{n-1}|\bra{\phi_i}A^T\ket{\hat{j}}|^2}{2(1+\frac{1}{2n})nd_{\mathrm{avg}}}\nonumber\\
= \frac{\sum_{i=1}^{\kappa-1}\norm{A\ket{\phi_i}}^2}{2(1+\frac{1}{2n})nd_{\mathrm{avg}}}.
\end{equation}
Since $\ket{\phi_i}\in \mathrm{row}(A)$ for each $i$, $\norm{A\ket{\phi_i}}^2\geq \sigma_{\min}(A)^2$. Thus 
\begin{equation}
\norm{(I\otimes \Pi)\ket{\psi_{\mathrm{init}}}}^2 
\geq \frac{(\kappa-1)\sigma_{\min}(A)^2}{2(1+\frac{1}{2n})nd_{\mathrm{avg}}}
=\frac{(\kappa-1)2\lambda_2(G)}{2(1+\frac{1}{2n})nd_{\mathrm{avg}}}.
\end{equation}
Here we used the fact that $G$ is assumed to be connected, the second smallest eigenvalue of $2L_G$ is the smallest non-zero eigenvalue, so $\lambda_2(2L_G)=\lambda_2(AA^T)=\sigma_{\min}(A)^2$.
\end{proof}

We remark that although this initial state lends itself well to analysis, it might be a particularly bad choice of an initial state, because it ensures lower weight on lower eigenvalue eigenstates of $G$, which may have high overlap with the $0$-eigenvalue eigenstates of $G(x)$. 

We next describe how we can construct the initial state. 

\begin{lemma}\label{lem:psi-init}
Let $\mathsf{S}$ be the cost of generating the stationary state of the graph $G$ 
$$\sum_{u\in V}\sqrt{\frac{d_G(u)}{d_{\mathrm{avg}}n}}\ket{u}.$$
Let $\mathsf{U}$ be the cost of implementing a step of the quantum walk on $G$, that is, generating, for any $u\in V$, a state of the form 
$$\sum_{v,\ell:(u,v,\ell)\in\overrightarrow{E}(G)}\sqrt{c(u,v,\ell)/d_G(u)}\ket{v,\ell}.$$
Then the map $\ket{0}\mapsto \ket{\psi_{\mathrm{init}}}$ can be implemented with error probability at most $\eps$ in time complexity $O((\mathsf{S}+\mathsf{U}+\log n)\log\frac{1}{\eps})$. 
\end{lemma}
\begin{proof}
We have
\begin{eqnarray}
\ket{\psi_j} &=& A^T\ket{\hat{j}} = \sum_{(u,v,\ell)\in\overrightarrow{E}(G)}\sqrt{c(u,v,\ell)}\frac{1}{\sqrt{n}}(e^{2\pi i ju/n}-e^{2\pi i jv/n})\ket{u,v,\ell}\\
&=& \frac{1}{\sqrt{n}}\sum_{u\in V}e^{2\pi i j u/n}\ket{u}\sum_{\substack{v,\ell:\\(u,v,\ell)\in\overrightarrow{E}(G)}}\sqrt{c(u,v,\ell)}\ket{v,\ell} - \frac{1}{\sqrt{n}}\sum_{u\in V}\ket{u}\sum_{\substack{v,\ell:\\(u,v,\ell)\in\overrightarrow{E}(G)}}e^{2\pi ij v/n}\sqrt{c(u,v,\ell)}\ket{v,\ell}.\nonumber
\end{eqnarray}
We first note that we can generate the state 
\begin{equation}
\sum_{u\in V}\sqrt{\frac{d_G(u)}{nd_{\mathrm{avg}}}}\ket{u}
\end{equation}
in cost $\mathsf{S}$ from which we can generate, for any $j\in[n]$,
\begin{equation}
\sum_{u\in V}e^{2\pi iju/n}\sqrt{\frac{d_G(u)}{nd_{\mathrm{avg}}}}\ket{u}
\end{equation}
using a generalized $Z_n^j$ gate, which performs the map $\ket{u}\mapsto e^{2\pi i uj/n}\ket{u}$ with complexity $O(\log n)$. From this, with one step of the quantum walk, we can get
\begin{equation}
\ket{\alpha_j}=\frac{1}{\sqrt{nd_{\mathrm{avg}}}}\sum_{u\in V}e^{2\pi iju/n}\ket{u}\sum_{v,\ell:(u,v,\ell)\in\overrightarrow{E}(G)}\sqrt{c(u,v,\ell)}\ket{v,\ell}
\end{equation}
in cost $\mathsf{U}$. The total cost of constructing $\ket{\alpha_j}$ is $O(\mathsf{S}+\mathsf{U}+\log n)$. 

Next, we can construct the state 
$$\sum_{u\in V}\sqrt{\frac{d_G(u)}{nd_{\mathrm{avg}}}}\ket{u}$$
in cost $\mathsf{S}$, from which we can construct
$$\frac{1}{\sqrt{nd_{\mathrm{avg}}}}\sum_{u\in V}\ket{u}\sum_{v,\ell:(u,v,\ell)\in\overrightarrow{E}(G)}\sqrt{c(u,v,\ell)}\ket{v,\ell}$$
in cost $\mathsf{U}$. Finally, for any $j\in[n]$, applying a generalized $Z_n^j$ gate on the second register, we get
$$\ket{\beta_j}=\frac{1}{\sqrt{nd_{\mathrm{avg}}}}\sum_{u\in V}\ket{u}\sum_{v,\ell:(u,v,\ell)\in\overrightarrow{E}(G)}e^{2\pi ijv/n}\sqrt{c(u,v,\ell)}\ket{v,\ell},$$
for total cost $O(\mathsf{U}+\mathsf{S}+\log n)$. One can now see that $\ket{\alpha_j}-\ket{\beta_j} = \frac{1}{\sqrt{d_{\mathrm{avg}}}}\ket{\psi_j}$. 

To construct $\ket{\psi_{\mathrm{init}}}$, generate the state:
$$\frac{1}{\sqrt{2(n-1)}}\ket{0}\sum_{j=1}^{n-1}\ket{j}\ket{\alpha_j}-\frac{1}{\sqrt{2(n-1)}}\ket{1}\sum_{j=1}^{n-1}\ket{j}\ket{\beta_j}.$$
This costs $O(\mathsf{S}+\mathsf{U}+\log n)$. Next, apply a Hadamard gate to the first register to get:
\begin{eqnarray}
&& \frac{1}{2\sqrt{n-1}}\ket{0}\sum_{j=1}^{n-1}\ket{j}(\ket{\alpha_j}-\ket{\beta_j})+\frac{1}{2\sqrt{n-1}}\ket{1}\sum_{j=1}^{n-1}\ket{j}(\ket{\alpha_j}+\ket{\beta_j})\nonumber\\
&=& \frac{1}{2\sqrt{(n-1)d_{\mathrm{avg}}}}\ket{0}\sum_{j=1}^{n-1}\ket{j}\ket{\psi_j}+\frac{1}{2\sqrt{n-1}}\ket{1}\sum_{j=1}^{n-1}\ket{j}(\ket{\alpha_j}+\ket{\beta_j}).
\end{eqnarray}
By Eq.~\eqref{eq:psi-init-norm}, we have $\norm{\sum_{j=1}^{n-1}\ket{j}\ket{\psi_j}}=\sqrt{2(1+1/(2n))nd_{\mathrm{avg}}}$, so the amplitude on the $\ket{0}$ part of the state is at least $\frac{1}{\sqrt{2}}$. Thus, we can measure the first register, and post select on measuring $\ket{0}$ to obtain $\ket{\psi_{\mathrm{init}}}$. With $\log\frac{1}{\eps}$ repetitions, we succeed with probability $1-\eps$.
\end{proof}

We can now give an upper bound on the complexity of deciding connectivity for any
family of parent graphs $G$, in terms of $\sf U$ and $\sf S$, in
\cref{thm:any-G}, below. We first note that it is reasonable to assume that
these costs should be low in many natural cases. The cost $\mathsf{U}$ is the
cost of implementing a step of a quantum walk on $G$, and note that $G$ is
input-independent, so as long as it is sufficiently structured, this shouldn't
be a particularly large cost. For example, if for any vertex in $G$, we can
efficiently query its degree, and its $i\tth$ neighbour for any $i$, then
$\mathsf{U}=O(\log n)$. Note that this is \emph{not} the same as assuming we
can efficiently query the $i\tth$ neighbour of a vertex in $G(x)$, which is not
an operation that we can easily implement in the edge-query input model.
Similarly, we might hope that $\mathsf{S}$ is also $O(\log n)$ in many cases
of interest. Indeed, whenever $G$ is $d$-regular, it's simply the cost of generating the uniform
superposition over all vertices.

\anyG*
\begin{proof}
By \cref{lem:psi-init}, the complexity of generating $\ket{\psi_{\mathrm{init}}}$ is $\mathsf{Init}=O(\mathsf{S}+\mathsf{U}+\log n)$, and by \cref{lem:eps}, the initial state has overlap at least $\varepsilon = \Omega\left(\frac{\kappa\lambda_2(G)}{nd_{\mathrm{avg}}}\right)$ with any unit vector in $\ker A(x)\cap \mathrm{row}(A)$. Plugging these values into the expression in \cref{thm:alg-template} gives (neglecting polylogarithmic factors)
\begin{equation*}
O\left(\frac{1}{\sqrt{\varepsilon}}\left(\mathsf{Init}+\sqrt{\frac{d_{\max}(G)}{\lambda}}\mathsf{U}\right)\right)
=\widetilde{O}\left(\sqrt{\frac{nd_{\mathrm{avg}}}{\kappa\lambda_2(G)}}\left(\mathsf{S}+\sqrt{\frac{d_{\max}(G)}{\lambda}}\mathsf{U}\right)\right).\qedhere
\end{equation*}
\end{proof}

\subsection{An Algorithm for Cayley Graphs}\label{sec:Cayley-graph}

When the parent graph $G$ is a Cayley graph for a finite Abelian group, we can
use the extra structure to construct an orthonormal basis of
$\mathrm{row}(A)$. We first define a Cayley graph. Let $\Gamma$ be a finite
Abelian group, and $S$ a symmetric subset of $\Gamma$, meaning that if $g\in
S$, then $-g\in S$. The Cayley graph $\mathsf{Cay}(\Gamma,S)$ is the graph
that has $\Gamma$ as its vertex set, and edge set $\{\{a,b\}: b-a\in S\}$.

For a positive integer $m$, let $\omega_m=e^{2\pi i/m}$. For an Abelian group
$\Gamma = \mathbb{Z}/m_1\mathbb{Z} \times \dots \times
\mathbb{Z}/m_k\mathbb{Z}$ and an element $g\in\Gamma$, we define the character
$\chi_g:G\rightarrow\mathbb{C}$ as the function $\chi_g(s) =
\omega_{m_1}^{g_1s_1}\dots\omega_{m_k}^{g_ks_k}$.

When $G=\mathsf{Cay}(\Gamma,S)$ is a Cayley graph, with $n=|\Gamma|$ and
$d=|S|$, it is easily verified that the eigenvectors of $L_G$ are exactly the
Fourier vectors: $$\ket{\hat{g}} =
\frac{1}{\sqrt{n}}\sum_{h\in\Gamma}\chi_g(h)\ket{h}.$$ For $g\neq 0$, these are also
the left-singular vectors of $A$. Most importantly, the vectors
$A^T\ket{\hat{g}}$ for $g\neq 0$ are an orthogonal basis of $\mathrm{row}(A)$
since they are proportional to the right-singular vectors of $A$. We define
\begin{eqnarray}
\ket{\psi_g}&=&A^T\ket{\hat{g}}=\sum_{u\in\Gamma}\sum_{v:v-u\in S}\sum_{h\in\Gamma}\frac{\chi_g(h)}{\sqrt{n}}\ket{u,v}(\bra{u}-\bra{v})\ket{h}=\sum_{u\in\Gamma}\sum_{v:v-u\in S}\frac{1}{\sqrt{n}}(\chi_g(u)-\chi_g(v))\ket{u,v}\label{eq:right-singular-vec1}\nonumber\\
&=&\sum_{u\in\Gamma}\frac{\chi_g(u)}{\sqrt{n}}\sum_{s\in S}(1-\chi_g(s))\ket{u,u+s}\label{eq:right-singular-vec2}.
\end{eqnarray}

We have $\norm{\ket{\psi_g}}^2=\norm{A^T\ket{\hat{g}}}^2=\lambda_g$, where
$\lambda_g$ is the eigenvalue of $L_G$ associated with $\ket{\hat{g}}$. In particular, $\lambda_g = d-\sum_{s\in S}\chi_g(s)$ (See for example
\cite{bollobas2013modern}).

We define
\begin{equation}
\ket{\psi_{\mathrm{init}}} = \frac{1}{\sqrt{n-1}}\sum_{g\in\Gamma\setminus\{0\}}\frac{1}{\sqrt{\lambda_g}}\ket{g}\ket{\psi_g}.\label{eq:init-cayley}
\end{equation}
We first lower bound the overlap of the initial state with the $0$-phase space of $U(P,x)$ when $G(x)$ is not connected.
\begin{lemma}\label{lem:cayley-eps}
Suppose $G(x)$ has at least $\kappa>1$ components. Let
$\ket{\psi_{\mathrm{init}}}$ be as in \cref{eq:init-cayley}, and let
$\{\ket{\phi_i}\}_{i=1}^{\kappa-1}$ be orthonormal 0-phase vectors of
$U(P,x)$ in $\mathrm{row}(A)$. Let
$\Pi=\sum_{i=1}^{\kappa-1}\ket{\phi_i}\bra{\phi_i}$. Then $$\norm{(I\otimes
\Pi)\ket{\psi_{\mathrm{init}}}}^2\geq \frac{\kappa - 1}{n-1}.$$
\end{lemma}
\begin{proof}
Let $\ket{\overline{\psi}_g}=\ket{\psi_g}/\sqrt{\lambda_g}$. Then $\{\ket{\overline{\psi}_g}\}_{g\neq 0}$ is an orthonormal basis for $\mathrm{row}(A)$.  
We have
\begin{equation}
\norm{(I\otimes \Pi)\ket{\psi_{\mathrm{init}}}}^2
=\frac{1}{n-1} \sum_{i=1}^{\kappa-1}\sum_{g\in\Gamma\setminus\{0\}}|\braket{\phi_i}{\overline{\psi}_g}|^2 = \frac{1}{n-1}\sum_{i=1}^{\kappa-1}1=\frac{\kappa-1}{n-1}.\qedhere
\end{equation}
\end{proof}

Next, we give an upper bound on the time complexity of constructing the initial state. We will use the following fact:
\begin{claim}[(See, for example, \cite{bollobas2013modern})]\label{claim:R-avg}
Let $G$ be any connected graph, with non-zero eigenvalues $\lambda_2,\dots,\lambda_n$. 
Then $R_{\mathrm{avg}}(G)=\frac{1}{n-1}\sum_{i=2}^n\frac{1}{\lambda_i}.$
\end{claim}

\begin{lemma}\label{lem:cayley-init}
Let $\mathsf{U}$ be the cost of generating the state $\frac{1}{\sqrt{d}}\sum_{s\in S}\ket{s}$. Let $\mathsf{\Lambda}$ be the cost of implementing, for $g\in\Gamma$, $\ket{g}\ket{0}\mapsto \ket{g}\ket{\lambda_g}$.  
Then the cost of generating the state $\ket{\psi_{\mathrm{init}}}$ as in \cref{eq:init-cayley} with success probability $1-\eps$ is 
$$O\left(\left(\sqrt{R_{\mathrm{avg}}(G)d}\left(\mathsf{U}+\log n\right)+\mathsf{\Lambda}{\sqrt{\frac{d}{\lambda_2(G)}}}\right)\log\frac{1}{\eps}\right).$$
\end{lemma}
\begin{proof}
The proof is similar to that of \cref{lem:psi-init}. 
Using a Fourier transform, we can generate the state 
\begin{equation}
\ket{g}\mapsto \frac{1}{\sqrt{n}}\sum_{u\in\Gamma}\chi_g(u)\ket{u}
\end{equation}
for any $g\in\Gamma$ in time $O(\log n)$. We can then generate $\sum_{s\in S}\frac{1}{\sqrt{d}}\ket{s}$ in time $\mathsf{U}$, and perform the map
\begin{equation}
\frac{1}{\sqrt{n}}\sum_{u\in\Gamma}\chi_g(u)\ket{u}\sum_{s\in S}\frac{1}{\sqrt{d}}\ket{s}\mapsto \frac{1}{\sqrt{dn}}\sum_{u\in\Gamma}\chi_g(u)\ket{u}\sum_{s\in S}\ket{u+s}=:\ket{\alpha_g}
\end{equation}
for a total complexity of $O(\log n+\mathsf{U})$ to generate $\ket{\alpha_g}$. 

Alternatively, we can use the generalized $Z_{\Gamma}^g$ gate, which maps $\ket{s}$ to $\chi_g(s)\ket{s}$ in time $O(\log n)$ to get
\begin{equation}
\frac{1}{\sqrt{nd}}\sum_{u\in\Gamma}\chi_g(u)\ket{u}\sum_{s\in S}\ket{s} \mapsto 
\frac{1}{\sqrt{nd}}\sum_{u\in\Gamma}\chi_g(u)\ket{u}\sum_{s\in S}\chi_g(s)\ket{s}
\mapsto \frac{1}{\sqrt{nd}}\sum_{u\in\Gamma}\chi_g(u)\ket{u}\sum_{s\in S}\chi_g(s)\ket{u+s}=:\ket{\beta_g}
\end{equation}
for a total complexity of $O(\log n+\mathsf{U})$ to generate $\ket{\beta_g}$. Observe that $\ket{\alpha_g}-\ket{\beta_g} = \frac{1}{\sqrt{d}}\ket{\psi_g}$.

Now to construct $\ket{\psi_{\mathrm{init}}}$, we first construct $\sum_{g\in\Gamma\setminus\{0\}}\frac{1}{\sqrt{\lambda_g}}\ket{g}$, as follows. We first generate:
\begin{equation}
\sum_{g\in\Gamma\setminus\{0\}}\frac{1}{\sqrt{n-1}}\ket{g}\ket{\lambda_g},
\end{equation}
in cost $\mathsf{\Lambda}$. Next, we map this to:
\begin{equation}
\sum_{g\in\Gamma\setminus\{0\}}\frac{1}{\sqrt{n-1}}\ket{g}\ket{\lambda_g}\left(\sqrt{\frac{\lambda_2(G)}{\lambda_g}}\ket{0}+\sqrt{1-\frac{\lambda_2(G)}{\lambda_g}}\ket{1}\right). 
\end{equation}
We can perform this map because for all $g\in\Gamma\setminus\{0\}$, $\lambda_2(G)/\lambda_g \leq 1$.
We uncompute $\ket{\lambda_g}$, and then do amplitude amplification on $\ket{0}$ in the last register to get the desired state. The squared amplitude on $\ket{0}$ is:
\begin{equation}
\norm{\sum_{g\in\Gamma\setminus\{0\}}\frac{1}{\sqrt{n-1}}\sqrt{\frac{\lambda_2(G)}{\lambda_g}}\ket{g}}^2 = \frac{\lambda_2(G)}{n-1}\sum_{g\in\Gamma\setminus\{0\}}\frac{1}{\lambda_g} = \lambda_2(G)R_{\mathrm{avg}}(G).
\end{equation}
So we can generate the normalized state 
\begin{equation}
\frac{1}{\sqrt{(n-1)R_{\mathrm{avg}}(G)}}\sum_{g\in\Gamma\setminus\{0\}}\frac{1}{\sqrt{\lambda_g}}\ket{g}
\end{equation}
with constant success probability in time complexity $O(\mathsf{\Lambda}/\sqrt{\lambda_2(G)R_{\mathrm{avg}}(G)})$. 

Next, we map this state to:
\begin{equation}
\mapsto \frac{1}{\sqrt{2(n-1)R_{\mathrm{avg}}}}\ket{0}\sum_{g\in \Gamma\setminus\{0\}}\frac{1}{\sqrt{\lambda_g}}\ket{g}\ket{\alpha_g} - \frac{1}{\sqrt{2(n-1)R_{\mathrm{avg}}}}\ket{1}\sum_{g\in\Gamma\setminus\{0\}}\frac{1}{\sqrt{\lambda_g}}\ket{g}\ket{\beta_g}
\end{equation}
at an additional cost of $O(\mathsf{U}+\log n)$. Next, we apply a Hadamard gate to the first qubit to get 
\begin{eqnarray}
&& \frac{1}{2\sqrt{(n-1)R_{\mathrm{avg}}}}\ket{0}\sum_{g\in\Gamma\setminus\{0\}}\frac{1}{\sqrt{\lambda_g}}\ket{g}(\ket{\alpha_g}-\ket{\beta_g})+\frac{1}{2\sqrt{(n-1)R_{\mathrm{avg}}}}\ket{1}\sum_{g\in\Gamma\setminus\{0\}}\frac{1}{\sqrt{\lambda_g}}\ket{g}(\ket{\alpha_g}+\ket{\beta_g})\nonumber\\
&=& \frac{1}{2\sqrt{(n-1)R_{\mathrm{avg}}d}}\ket{0}\sum_{g\in\Gamma\setminus\{0\}}\frac{1}{\sqrt{\lambda_g}}\ket{g}\ket{\psi_g}+\frac{1}{2\sqrt{(n-1)R_{\mathrm{avg}}}}\ket{1}\sum_{g\in\Gamma\setminus\{0\}}\frac{1}{\sqrt{\lambda_g}}\ket{g}(\ket{\alpha_g}+\ket{\beta_g}).
\end{eqnarray}
The total cost to make one copy of this state with constant success probability is $O(\mathsf{U}+\log n+\mathsf{\Lambda}/\sqrt{\lambda_2(G)R_{\mathrm{avg}}})$. We can then get $\ket{\psi_{\mathrm{init}}}$ by doing amplitude amplification on the $\ket{0}$ part of this state. The amplitude on the $\ket{0}$ part of the state is given by 
$$\norm{\frac{1}{2\sqrt{(n-1)R_{\mathrm{avg}}d}}\sum_{g\in\Gamma\setminus\{0\}}\frac{1}{\sqrt{\lambda_g}}\ket{g}\ket{\psi_g}}
=\frac{1}{2\sqrt{R_{\mathrm{avg}}d}},$$
so using $O\left(\sqrt{R_{\mathrm{avg}}(G)d}\right)$ rounds of amplitude amplification is sufficient to generate $\ket{\psi_{\mathrm{init}}}$ with constant probability, for a total cost of (neglecting constants):
\begin{equation}
\sqrt{R_{\mathrm{avg}}(G)d}\left(\mathsf{U}+\log n + \frac{\sf \Lambda}{\sqrt{\lambda_2(G) R_{\mathrm{avg}}(G)}}\right) = \sqrt{R_{\mathrm{avg}}(G)d}\left(\mathsf{U}+\log n\right) + {\sf \Lambda}\sqrt{\frac{d}{\lambda_2(G)}}.
\end{equation}
We can amplify this to success probability $1-\eps$ at the cost of a $\log(1/\eps)$ multiplicative factor.
\end{proof}

\Cayleygraph*
\begin{proof}
Combining \cref{lem:cayley-init} and \cref{lem:cayley-eps} with \cref{thm:alg-template}, we get complexity:
\begin{equation}
\widetilde{O}\left(\sqrt{ndR_{\mathrm{avg}}(G)/\kappa}\mathsf{U}+\sqrt{\frac{nd}{\lambda\kappa}}\mathsf{U}+\sqrt{\frac{nd}{\lambda_2(G)\kappa}}\mathsf{\Lambda}\right).
\end{equation}
First, note that $\lambda_2(G(x))\leq \lambda_2(G)$. Thus
if $\lambda > \lambda_2(G)$, then there is no $x$ such that $\lambda_2(G(x))\leq \lambda$, so $X$ is empty. Thus, we can assume $\lambda \leq \lambda_2(G)$. 
By \cref{claim:R-avg}, we can see that $R_{\mathrm{avg}}(G)\leq \frac{1}{\lambda_2(G)}\leq \frac{1}{\lambda}$. The claim follows.
\end{proof}

We now look at specific examples where it is particularly efficient to compute
$\lambda_g$, as well as prepare a step of the walk $\sum_{s\in S}\ket{s}$. We
first consider the complete graph on $n$ vertices, in which $\Gamma =
\mathbb{Z}_n$, and $S=\Gamma\setminus\{0\}$.

\begin{corollary}\label{cor:complete-graph}
Fix any $\lambda>0$, and integer $\kappa>1$ and let $G$ be the complete graph.
Let $X\subseteq E(G)$ be such that for all $x\in X$, either
$\lambda_2(G(x))\geq \lambda$, or $G(x)$ has at least $\kappa$ components.
Then $\textsc{conn}_{G,X}$ can be solved in bounded error in time
\begin{equation*}
\widetilde{O}\left(\frac{n}{\sqrt{\kappa\lambda}}\right).
\end{equation*}
\end{corollary}
\begin{proof}
It is easily verified that for $G$ a complete graph, $L_G=(n-1)I-(J-I)$, where
$J$ is the all-ones matrix, so the eigenvalues consist of a single 0, and $n$
with multiplicity $n-1$ --- that is, all non-zero eigenvalues are $n$. Thus,
the mapping $\ket{g}\mapsto \ket{g}\ket{\lambda_g}=\ket{g}\ket{n}$ can be
implemented trivially in $O(\log n)$ complexity.

Next, we can generate the state $\sum_{s\in
S}\frac{1}{\sqrt{d}}\ket{s}=\sum_{s=1}^{n-1}\frac{1}{\sqrt{n-1}}\ket{s}$ in
complexity $\mathsf{S}=O(\log n)$. Then the result follows from 
\cref{thm:Cayley-graph}.
\end{proof}

Next, we consider the Boolean hypercube, in which $\Gamma =\mathbb{Z}_2^d$, so $n=2^d$, and $S=\{e_i\}_{i=1}^d$, where $e_i$ is 0 everywhere except the $i\tth$ entry, which is 1. 

\begin{corollary}\label{cor:hypercube}
Fix any $\lambda>0$, and integer $\kappa>1$ and let $G$ be the Boolean hypercube on $n=2^d$ vertices. Let $X\subseteq E(G)$ be such that for all $x\in X$, either $\lambda_2(G(x))\geq \lambda$, or $G(x)$ has at least $\kappa$ components. Then $\textsc{conn}_{G,X}$ can be solved in bounded error in time
\begin{equation*}
\widetilde{O}\left(\sqrt{\frac{n}{\kappa\lambda}}\right).
\end{equation*}
\end{corollary}
\begin{proof}
The eigenvalues of the Boolean hypercube are well known to be
$\lambda_g=2|g|$, for $g\in\mathbb{Z}_2^d$, where $|g|$ denotes the Hamming
weight of $g$. Thus, the map $\ket{g}\mapsto\ket{g}\ket{\lambda_g}$ can be
implemented in cost $O(\log n)$. Finally, the state $\sum_{s\in
S}\ket{s}=\sum_{i=1}^d\ket{e_i}$ can be generated in time $O(\log n)$.
Then by \cref{thm:Cayley-graph}, the time complexity is (neglecting polylog
factors):
\begin{equation}
\sqrt{\frac{nd}{\kappa\lambda}}=\widetilde{O}\left(\sqrt{\frac{n}{\kappa\lambda}}\right).\qedhere
\end{equation}
\end{proof}

\subsection{Estimating the connectivity when $G$ is a complete graph}
\label{sec:estimating-connectivity}

For the remainder of this section, let $G$ be the complete graph on $n$ vertices, $K_n$. In that case, we can not only decide if $G(x)$ is connected, but estimate $\lambda_2(G(x))$. 
The idea is to relate the smallest phase of $U(P,x)$ on $\mathrm{row}(A)$ to $\lambda_2(G(x))$, and estimate this value using quantum phase estimation. 

Let $\Delta(U(P,x))$ denote the smallest nonzero phase of $U(P,x)$, which we want to estimate.  We will shortly show that there is a vector $\ket{u}$ in $\mathrm{row}(A)$ in the $\pm\Delta(U(P,x))$-phase space of $U(P,x)$. Then, if $\{\ket{\psi_i}\}_{i=1}^{n-1}$ is an orthonormal basis for $\mathrm{row}(A)$, applying quantum phase estimation on the second register of 
\begin{equation}
\ket{\psi_{\mathrm{init}}}=\frac{1}{\sqrt{n-1}}\sum_{i=1}^{n-1}\ket{i}\ket{\psi_i}
\end{equation}
to some constant precision, and then applying amplitude estimation to determine if there is amplitude at least $\frac{1}{\sqrt{2(n-1)}}$ on phases less than $1/2$, we can distinguish between the case in which, say, $\Delta(U(P,x))\leq 1/3$, and $\Delta(U(P,x))\geq 2/3$. To get an accurate estimate of $\Delta(U(P,x))$, we will make repeated calls to such a phase-estimation-followed-by-amplitude-estimation subroutine, reducing the size of the interval where $\Delta(U(P,x))$ sits at every iteration. We will first describe the connection between $\Delta(U(P,x))$ and $\lambda_2(G(x))$, and then formally present the algorithm for estimating $\Delta(U(P,x))$.

\paragraph{Connection between $\lambda_2(G(x))$ and $\Delta(U(P,x))$} We use the following theorem, relating the phases of the product of two reflections $U=(2\Pi_A-I)(2\Pi_B-1)$ to the singular values of its discriminant, defined $\Pi_A\Pi_B$.
\begin{theorem}[\cite{sze04}]\label{thm:sze}
Let $\Pi_A=\sum_{i=1}^a\ket{\alpha_i}\bra{\alpha_i}$ and $\Pi_B=\sum_{i=1}^b\ket{\beta_i}\bra{\beta_i}$ be orthogonal projectors into some subspaces of the same inner product space, and define $U=(2\Pi_A-I)(2\Pi_B-I)$. Let $D=\Pi_A\Pi_B$ be the discriminant of $U$, and suppose it has singular value decomposition $D=\sum_{j=1}^r\cos\theta_j \ket{u_j}\bra{v_j}$ with $\theta_j\in [0,\pi/2)$. Then $U$ has 1-eigenspace $(A\cap B)\oplus (A^\bot\cap B^\bot)$ and $(-1)$-eigenspace $(A\cap B^\bot)\oplus (A^\bot\cap B)$. The only other eigenvalues of $U$ are exactly $\{e^{\pm 2i\theta_j}\}_{j=1}^r$, and for each $j$, the $e^{2i\theta_j}$- and $e^{-2i\theta_j}$-eigenvectors are, respectively, $\ket{\theta_j^+}=\ket{v_j}-e^{i\theta_j}\ket{u_j}$ and $\ket{\theta_j^-}=\ket{v_j}-e^{-i\theta_j}\ket{u_j}$.
\end{theorem}

We derive several consequences of this theorem, and specialize them to our particular setting.

\begin{lemma}\label{lem:sze1}
Let  $U=(2\Pi_A-I)(2\Pi_B-I)$ and $D=\Pi_A\Pi_B$ be its discriminant. Then $\Delta(-U)=2\sin^{-1}(\sigma_{\min}(D))$. Moreover, when $G$ is a complete graph on $n$ vertices, we have for any $x$, $\lambda_2(G(x))=n\sin^2(\Delta(U(P,x))/2)$.
\end{lemma}
\begin{proof}
Assume without loss of generality that $\sigma_{\min}(D)=\cos\theta_1$. Since
$\cos$ is a decreasing function in the interval $[0,\frac{\pi}{2}]$, it
follows that $\theta_1\geq\theta_j$ $\forall j>1$. By \cref{thm:sze}, the
spectrum of  $U$ outside of its $(\pm 1)$-eigenspace is $\{e^{\pm2i\theta_j}\}_{j=1}^r$, hence the biggest
eigenphase smaller than $\pi$ in absolute value must be $2\theta_1$. In other words,
$\pi-2\theta_1$ is the phase gap of $-U$.

Let $\alpha$ be the complementary angle to $\theta_1$, i.e. $\pi/2=\alpha+\theta_1$. Then $2\alpha=\pi-2\theta_1=\Delta(-U)$. but $\cos(\theta_1)=\sin(\alpha)$, which means $\sigma_{\min}(D)=\sin(\alpha)=\sin(\Delta(-U)/2)$. The first result follows from applying the arcsine function on both sides of the last equality.

In particular, if we take $U=-U(P,x)=(2\Pi_{\mathrm{row}(A)}-I)(2\Pi_{H(x)}-I)$, then we have 
$$D=\Pi_{\mathrm{row}(A)}\Pi_{H(x)} = A^+A\Pi_{H(x)}=A^+A(x),$$
with $\sin(\Delta(U(P,x))/2) = \sigma_{\min}(D)$.

We have $AA^T = 2L_G$, for $G$ the complete graph. It is known that when $G$ is a complete graph $K_n$, $L_G$ has 0-eigenspace spanned by the uniform vector $\ket{\mu}=\frac{1}{\sqrt{n}}\sum_{u\in [n]}\ket{u}$, and moreover, we have
$$L_G= (n-1)I - (J-I) = nI - J = nI - n\ket{\mu}\bra{\mu} = n\sum_{i=1}^{n-1}\ket{b_i}\bra{b_i},$$
where $\{\ket{b_i}\}_{i=1}^{n-1}$ is any orthonormal basis for $\mathrm{span}\{\ket{\mu}\}^{\bot}$, which is $\mathrm{col}(A)$. This implies that for some orthonormal basis for $\mathrm{row}(A)$,  $\{\ket{\psi_i}\}_{i=1}^{n-1}$:
\begin{equation*}
A=\sum_{i=1}^{n-1}\sqrt{2n}\ket{b_i}\bra{\psi_i}.
\end{equation*}
Thus 
\begin{equation*}
A^+=\sum_{i=1}^{n-1}\frac{1}{\sqrt{2n}}\ket{\psi_i}\bra{b_i}.
\end{equation*}
Next, note that since $L_G\ket{\mu}=0$ for any $G$, and $2L_G=A(x)A(x)^T$, we
have $A(x)^T\ket{\mu}=0$, so the columnspace of $A(x)$ is in $\mathrm{span}\{\ket{\mu}\}^\bot$, and in particular, if $G(x)$ is connected, it's exactly $\mathrm{span}\{\ket{\mu}\}^\bot$.
The basis $\{\ket{b_i}\}_{i=1}^{n-1}$ can be chosen to be any basis of
$\ket{\mu}^\bot$, so let's choose it to be the
right singular basis of $A(x)$. That is, there exist
$\ket{\phi_i}$ and $\sigma_i$ such that
$$A(x)=\sum_{i=1}^{n-1}\sigma_i\ket{b_i}\bra{\phi_i}$$ is a singular value
decomposition for $A(x)$. Then we have:
\begin{equation*}
D = A^+A(x) = \sum_{i=1}^{n-1}\frac{\sigma_i}{\sqrt{2n}}\ket{\psi_i}\bra{\phi_i}.
\end{equation*}
Since $2L_{G(x)}=A(x)A(x)^T$, the $\sigma_i$ are just the square roots of twice the nonzero eigenvalues $\lambda_2,\dots,\lambda_n$ of $L_{G(x)}$, so the singular values of $D$ are
\begin{equation*}
\left\{ \sqrt{\frac{2\lambda_2}{2n}},\dots,\sqrt{\frac{2\lambda_n}{2n}} \right\}.
\end{equation*}
We conclude that $\sigma_{\min}(D)=\sqrt{\frac{\lambda_2(G(x))}{n}}$, which, combined with $\sigma_{\min}(D)=\sin(\Delta(U(P,x))/2)$, gives $\lambda_2(G(x))=n\sin^2(\Delta(U(P,x))/2)$.
\end{proof}

Another consequence of \cref{thm:sze} is the following, which allows us to restrict our attention to $\mathrm{row}(A)$ in searching for the smallest phase of $U(P,x)$:
\begin{lemma}\label{lem:sze2}
Let $U=(2\Pi_A-I)(2\Pi_B-I)$, and let $\ket{\Delta_+}$ be a $\Delta(U)$-phase eigenvector of $U$, and $\ket{\Delta_-}$ a $(-\Delta(U))$-phase eigenvector of $U$. Then there exists a a vector $\ket{u}$ in the support of $A$ such that $\ket{u}\in\mathrm{span}\{\ket{\Delta_+},\ket{\Delta_-}\}$. In particular, if $\ket{\Delta_\pm}$ are $\pm\Delta(U(P,x))$-phase eigenvectors of $U(P,x)$, then there exists a vector $\ket{u}$ in $\mathrm{row}(A)$ such that $\ket{u}\in \mathrm{span}\{\ket{\Delta_+},\ket{\Delta_-}\}$. 
\end{lemma}
\begin{proof}
Let $\theta_j$, $\ket{u_j}$, $\ket{\theta_j^+}$ and $\ket{\theta_j^-}$ be as in \cref{thm:sze}, so in particular, $\ket{u_j}$ is in the support of $\Pi_A$. 
Note that for any $j$, we have 
\begin{equation}
\ket{u_j} = \frac{1}{2i\sin\theta_j}\left(\ket{\theta_j^+} - \ket{\theta_j^-}\right).
\end{equation}
In particular, this is true for the $j$ such that $\theta_j$ is minimized, i.e. such that $\theta_j=\Delta(U)$. The statement follows. 
\end{proof}

\paragraph{Algorithm for estimating $\Delta(U(P,x))$} We will actually estimate the value $\tau=\Delta(U(P,x))/\pi$, getting an estimate in $[0,1]$, which we will then transform into an estimate of $\lambda_2(G(x))$. At every iteration,
$c$ will denote a lower bound for $\tau$ and $C$ will denote the current
upper bound. At the beginning of the algorithm we have
$c=0$, $C=1$, and every iteration will result in updating either $C$ or $c$ in
such a manner that the new interval for $\tau$ is reduce by a
fraction of $2/3$. The algorithm is described in \cref{alg:est-alg}.

\begin{algorithm}\label{alg:est-alg} To begin, let $c=0$ and $C=1$.
\begin{enumerate}
\item Set $\varphi=\frac{C-c}{3}$, $\eps=\frac{1}{\sqrt{2n}}$, $\delta=c+\varphi$.
\item For $j=1,\dots, 4\log (n/\varepsilon)$:
\begin{enumerate}
\item Prepare $\sum_{i=1}^{n-1}\frac{1}{\sqrt{n-1}}\ket{i}\ket{\psi_i}\ket{0}_C\ket{0}_P$.
\item Perform the gapped phase estimation algorithm $GPE(\varphi,\eps,\delta)$ of \cref{thm:gapped-phase-estimation} applying $U(P,x)$ on the second register.
\item Use amplitude estimation (see \cref{thm:amplitude-estimation}) to distinguish between the case when the amplitude on $\ket{0}_C$ is $\geq \frac{1}{\sqrt{n}}$, in which case output ``$a_j=0$'', and the case where the ampltiude is $\leq \frac{1}{\sqrt{2n}}$, in which case, output ``$a_j=1$''.
\end{enumerate}
\item Compute $\tilde{a}=Maj(a_1,\dots,a_{4\log (n/\varepsilon)})$.  If the result is $0$, set $C= \delta+\varphi$. If the result is $1$, set $c=\delta$. If $C-c\leq 2\varepsilon c$, then output $n\sin^2\left(\frac{\pi(C+c)}{4}\right)$. Otherwise, return to Step 1.
\end{enumerate}
\end{algorithm}

\paragraph{Analysis of the algorithm} 

We say an iteration of the algorithm \emph{succeeds} if $\tilde a = Maj(a_1,\dots,a_{4\log(n/\varepsilon)})$ correctly indicates whether the amplitude on $\ket{0}_C$ is $\geq \frac{1}{\sqrt{n}}$ or $\leq \frac{1}{\sqrt{2n}}$. This happens with probability $\Omega(1-(\varepsilon/n)^4)$. Since we will shortly see that the algorithm runs for at most $\tO\left(\frac{n}{\varepsilon\sqrt{\lambda_2(G(x))}}\right)\leq \tO\left(\frac{n^2}{\varepsilon}\right)$ steps, the probability that every iteration succeeds is at least 
\begin{equation}
\Omega\left( \left(1-(\varepsilon/n)^4\right)^{(n/\varepsilon)^2} \right) = \Omega\left(1-(\varepsilon/n)^4(n/\varepsilon)^2\right)=\Omega\left(1-(\varepsilon/n)^2\right).
\end{equation}
It is therefore reasonable to assume that every iteration succeeds, since this happens with high probability.
We first prove that if every iteration succeeds, throughout the algorithm we have $\tau=\Delta(U(P,x))/\pi \in [c,C]$. 

\begin{lemma}
Let $\tau=\Delta(U(P,x))/\pi$.  For any $\varphi$ and $\delta$, if $\tau\geq \delta+\varphi$, applying $GPE(\varphi,\epsilon,\delta)$ to $\ket{\psi_{\mathrm{init}}}$ results in a state with amplitude at most $\frac{1}{\sqrt{2n}}$ on $\ket{0}_C$ in register $C$; and if $\tau\leq \delta$, this results in a state with amplitude at least $\frac{1}{\sqrt{n}}$ on $\ket{0}_C$ in register $C$. Thus, if every iteration succeeds, at every iteration, we have $\tau\in [c,C]$. 
\end{lemma}
\begin{proof}
First, suppose $\tau\geq \delta+\varphi$. By \cref{lem:kappa}, when $G(x)$ is connected there is no vector in $\mathrm{row}(A)$ in the $1$-eigenspace of $U(P,x)$, so the 1-eigenspace of $U(P,x)$ is contained in $\mathrm{ker}(A)\cap H(x)\subseteq \ker A$. Thus each $\ket{u_j}$ is in the span of $e^{i\pi\theta}$-eigenvectors of $U(P,x)$ with $|\theta|\geq \delta+\varphi$. 
Thus, applying $GPE(\varphi,\epsilon,\delta)$ will map each $\ket{u_j}_R\ket{0}_C\ket{0}_P$ to a state $\beta_0\ket{0}_C\ket{\gamma_0}_{PR}+\beta_1\ket{1}_C\ket{\gamma_1}_{PR}$ such that $|\beta_0|\leq \eps$. Then, by linearity, the total amplitude on $\ket{0}_C$ in register $C$ will be at most $\epsilon=\frac{1}{\sqrt{2n}}$. 

On the other hand, suppose $\tau\leq \delta$. By \cref{lem:sze2}, there exists a vector $\ket{u_1}\in\mathrm{row}(A)$ such that $\ket{u_1}$ is in the span of the $e^{\pm i\pi\tau}$-eigenvectors of $U(P,x)$. Applying $GPE(\varphi,\epsilon,\delta)$ will map $\ket{u_1}_R\ket{0}_C\ket{0}_P$ to a state $\beta_0\ket{0}_C\ket{\gamma_0}_{PR}+\beta_1\ket{1}_C\ket{\gamma_1}_{PR}$ such that $|\beta_1|\leq \eps$, so $|\beta_0|\geq \sqrt{1-\eps^2}$. 
Let $\ket{u_2},\dots,\ket{u_{n-1}}$ be any orthonormal set such that $\ket{u_1},\dots,\ket{u_{n-1}}$ is an orthonormal basis for $\mathrm{row}(A)$. 
Then there exists some (unknown) orthonormal set $\{\ket{\tilde j}\}_{j=1}^{n-1}$ such that 
\begin{equation}
\ket{\psi_{\mathrm{init}}}=\frac{1}{\sqrt{n-1}}\sum_{j=1}^{n-1}\ket{\tilde{j}}\ket{u_j}\ket{0}_C\ket{0}_P.
\end{equation}
So after applying $GPE(\varphi,\epsilon,\delta)$ to $\ket{\psi_{\mathrm{init}}}$, the amplitude on $\ket{0}_C$ in register $C$ will be at least $\sqrt{\frac{1-\eps^2}{n-1}}\geq \frac{1}{\sqrt{n}}$. This proves the first part of the statement. 

By \cref{cor:amplitude-estimation}, we can distinguish the case when the amplitude on $\ket{0}_C$ is at least $\frac{1}{\sqrt{n}}$ or at most $\frac{1}{\sqrt{2n}}$ with
bounded error using $O(\frac{\sqrt{p_0}}{p_0-p_1})=O(\sqrt{n})$ calls to
$GPE(\varphi,\frac{1}{\sqrt{2n}},\delta)$ where $p_0:=\frac{1}{n}$
and $p_1:=\frac{1}{2n}$. Thus, by repeating the procedure $4\log(n/\varepsilon)$ times and taking the majority, we succeed at every iteration with high probability.

We now prove by induction that we always have $\tau\in [c,C]$, as long as every iteration succeeds. At the beginning of the first iteration, we have $[c,C]=[0,1]$. Since $\Delta(U(P,x))\in [0,\pi]$, $\tau\in [0,1]$. Next, suppose in some arbitrary iteration, we have $\tau\in [c,C]$. If $\tau\geq \delta+\varphi$, then there will be amplitude at most $\frac{1}{\sqrt{2n}}$ on $\ket{0}_C$, and assuming the iteration succeeds, we will have $\tilde{a}=1$. In that case, we will set $c=\delta\leq \delta+\varphi\leq \tau$, so we will still have $\tau\in [c,C]$. If $\tau\leq \delta$, then there will be amplitude at least $\frac{1}{\sqrt{n}}$ on $\ket{0}_C$, and assuming the iteration succeeds, we will have $\tilde{a}=0$. In that case, we will set $C=\delta+\varphi\geq \delta\geq \tau$, so we will still have $\tau\in [c,C]$. 

We finally consider what happens if $\delta\leq \tau\leq \delta +\varphi$. In that case, there is no guarantee on the output of amplitude estimation; it can either output 0 or 1. However, we can still use the result to update our bounds for $\tau$. If we get $\tilde a=1$, and set $c=\delta$, we have $\delta\leq \tau$, so $\tau\in [c,C]$. If we get $\tilde a=0$, and set $C=\delta+\varphi$, we have $\delta+\varphi\geq \tau$, so $\tau\in [c,C]$. 
\end{proof}

Next, we analyze the running time of \cref{alg:est-alg}. 

\begin{theorem}
With probability $\Omega(1-(\varepsilon/n)^2)$, \cref{alg:est-alg} will terminate after time $\widetilde{O}\left(\frac{n}{\varepsilon\sqrt{\lambda_2(G(x))}}\right)$.
\end{theorem}
\begin{proof}
With probability $\Omega(1-(\varepsilon/n)^2)$, each of the first $(n/\varepsilon)^2\geq \widetilde{O}\left(\frac{n}{\varepsilon\sqrt{\lambda_2(G(x))}}\right)$ iterations of the algorithm will succeed, so we assume this to be the case. 
We first bound the number of (successful) iterations before the algorithm terminates. 
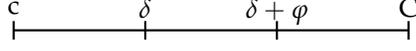
\begin{figure}[h]
\centering
\begin{tikzpicture}[scale=7]
    
\draw[-, thick] (0,0) -- (.75,0);
    \draw[thick] (0,0.5pt) -- (0,-0.5pt);
    \draw[thick] (0,1.2pt) node{c};
    \draw[thick] (0.25,0.5pt) -- (0.25,-0.5pt);
    \draw[thick] (0.25,1.2pt) node{$\delta$};
            \draw[thick] (0.5,0.5pt) -- (0.5,-0.5pt);
    \draw[thick] (0.5,1pt) node{$\delta+\varphi$};
    \draw[thick] (.75,0.5pt) -- (.75,-0.5pt);
    \draw[thick] (.75,1.2pt) node{C};
\end{tikzpicture}
\caption{The interval $[c,C]$}
\label{ConnEstimation}
\end{figure}
 \cref{ConnEstimation} shows the interval $[c,C]$, which represents the algorithm's current state of knowledge of where $\tau=\Delta(U(P,x))/\pi$ lies. The values $\delta$ and $\delta+\varphi$ are at $1/3$ and $2/3$ of the interval, respectively.  
it is not difficult to convince
ourselves that at every iteration, the interval $[c,C]$ will be $2/3$
the size it had in the previous iteration. Hence, since the interval initially has length 1, after $k$ iterations, we will have
an interval of size $\left(\frac{2}{3}\right)^k$. The execution terminates when the interval becomes sufficiently small. Specifically, let $T$ be the smallest integer such that $T\geq\frac{\log\frac{2}{\tau\varepsilon}}{\log\frac{3}{2}}$, and let $[c,C]$ be the interval after $T$ steps, so $C-c=(2/3)^T\leq \tau\varepsilon/2$. Suppose $(2/3)^T=C-c\geq 2\varepsilon c$, so $\tau\geq 4c$. This implies that $C\geq 4c$, so $C/c\geq 4$. We will argue that this is a contradiction. 

First, suppose $c=0$. That means that 
$$\tau\leq C=(2/3)^T\leq \tau\varepsilon/2\leq \tau/2,$$
which is a contradiction, since $\tau>0$. Thus, we must have $c>0$. Consider the first setting of $c$ and $C$ such that $c\neq 0$. Since the previous value of $c$ was 0, we set the new value as $c=\delta = 0+\varphi=(C-0)/3=C/3$, so the ratio $C/c$ satisfies $C/c=3$. This ratio can only decrease, because subsequent steps either decrease $C$, or increase $c$. Thus, after T steps, $C/c\leq 3$. Thus, $C-c\geq 2\varepsilon c$ leads to a contradiction, so we can conclude that after $T$ steps, $C-c\leq 2\varepsilon c$, so the algorithm terminates in at most $T$ steps.  

We can now analyze the total running time by adding up the cost of all iterations. 
Step 1 of the algorithm is defining the variables $\varphi=\frac{C-c}{4}$, $\delta=c+\varphi$ and
$\eps=\frac{1}{\sqrt{2n}}$, which will contribute negligibly to the complexity. 

Step 2(a) begins by constructing the initial state $\ket{\psi_{\mathrm{init}}}=\sum_{i=1}^{n-1}
\frac{1}{\sqrt{n-1}}\ket{i}\ket{\psi_i}\ket{0}_P\ket{0}_C$ where
$\{\ket{\psi_i}\}_{i=1}^{n-1}$ is a basis of $\mathrm{row}(A)$. Because
$G=K_n$ can be seen as a particularly simple kind of Cayley graph with group
$\Gamma=\mathbb{Z}/n\mathbb{Z}$ and $S=\Gamma\setminus\{0\}$, we can use the construction of
\cref{sec:Cayley-graph} to generate  $\ket{\psi_{\mathrm{init}}}$. In fact,
combining the remarks in the proof of \cref{cor:complete-graph} with 
 \cref{lem:cayley-init} it follows that this state can be constructed in time $O(\log n
\log \frac{n}{\varepsilon})$ with success probability $1-(\varepsilon/n)^4$. 

Step 2(b) consists of applying the unitary procedure $GPE(\varphi,\eps,\delta)$
described in \cref{thm:gapped-phase-estimation} on the last three registers
with $\varphi$, $\eps=\frac{1}{\sqrt{2n}}$, $\delta$ defined in Step 1. By \cref{thm:gapped-phase-estimation}, this makes $O(\varphi^{-1}\log\eps^{-1})=O(\varphi^{-1}\log n)$ calls to $U(P,x)$, for a total query complexity of $O(\varphi^{-1}\log^2n)$.

Step 2(c) then uses amplitude estimation, repeating Steps 2(a) and 2(b) $O(\sqrt{n})$ times, by \cref{cor:amplitude-estimation}. 
Let $\varphi^{(i)}$ denote the value of $\varphi$ at the $i\tth$ iteration of the algorithm. Neglecting polylog$(n/\varepsilon)$ factors, the running time of the $i\tth$ iteration is 
\begin{equation}
Q_i:=\frac{\sqrt{n}}{\varphi^{(i)}}.
\end{equation}
During the $i\tth$ iteration, we begin with $C-c=(2/3)^{i-1}$, and so $\varphi^{(i)}=\frac{1}{3}(2/3)^{i-1}$. Thus, we can compute the total complexity of the algorithm as (neglecting polylogarithmic factors):
\begin{equation}
\sum_{i=1}^TQ_i=\sqrt{n}\sum_{i=1}^{T}3(3/2)^{i-1} 
=3\sqrt{n}\frac{(3/2)^T-1}{3/2-1}
=\widetilde{O}\left(\sqrt{n}(3/2)^{\frac{\log(2/(\tau\varepsilon))}{\log(3/2)}}\right)
=\widetilde{O}\left(\frac{\sqrt{n}}{\tau\varepsilon}\right).
\end{equation}
By \cref{lem:sze1}, we have $\lambda_2(G(x))/n = \sin^2(\Delta(U(P,x))/2)\leq \Delta(U(P,x))^2/4$, so 
Filling in $\tau=\Delta(U(P,x))/\pi \geq \sqrt{\lambda_2(G(x))/(2n)}$, we get a total query complexity of $\widetilde{O}\left(\frac{n}{\varepsilon\sqrt{\lambda_2(G(x))}}\right)$.  
\end{proof}

Finally, we prove that the algorithm outputs an estimate that is within $\varepsilon$ multiplicative error of $\lambda_2(G(x))$.

\begin{theorem}[Correctness]
With probability at least $\Omega(1-(\varepsilon/n)^2)$, \cref{alg:est-alg} outputs an estimate $\tilde\lambda$ such that $\abs{\lambda_2(G(x))-\tilde\lambda}\leq \frac{\pi^23}{4}\varepsilon\lambda_2(G(x))$.
\end{theorem}
\begin{proof}
We will assume that all iterations succeed, which happens with probability at least $\Omega(1-(\varepsilon/n)^2)$.
Then the algorithm outputs $\tilde{\lambda}:=n\sin^2\left(\frac{\pi(C+c)}{4}\right)$ for some $c$ and $C$ such that $c\leq \tau\leq C$, and $C-c\leq 2\varepsilon c \leq 2\varepsilon \tau$. Using $\tau=\Delta(U(P,x))/\pi$ and $\lambda_2(G(x))=n\sin^2\left(\Delta(U(P,x))/2\right)$, we have:
\begin{equation}
\abs{\lambda_2(G(x))-\tilde\lambda} = \abs{n\sin^2\left(\pi\tau/2\right) - n\sin^2\left(\pi(C+c)/4\right)}.\label{eq:diff}
\end{equation}
From $c\leq \tau\leq C$ and $C-c\leq 2\varepsilon\tau$, we have 
\begin{equation}
\abs{\frac{\pi(C+c)}{4}-\frac{\pi\tau}{2}}\leq \frac{\pi\varepsilon\tau}{2}.
\end{equation}
Let $\delta = \pi(C+c)/4 - \pi\tau/2$, so $\pi(C+c)/4 = \pi\tau/2+\delta$. Then we have:
\begin{eqnarray}
\abs{\sin^2(\pi\tau/2) - \sin^2(\pi\tau/2 + \delta)}
& = & \abs{\frac{1-\cos(\pi\tau)}{2} - \frac{1-\cos(\pi\tau+2\delta)}{2}}\nonumber\\
&=& \frac{1}{2}\abs{\cos(\pi\tau+2\delta) - \cos(\pi\tau)}\nonumber\\
&=& \abs{\sin(\pi\tau+\delta)\sin(-\delta)}\nonumber\\
&\leq & \abs{\delta(\pi\tau+\delta)}
\; \leq \; \pi^2\tau^2\frac{\varepsilon}{2}\left(1+\frac{\varepsilon}{2}\right)
\;\leq \; \frac{3\varepsilon}{4}\pi^2\tau^2
\end{eqnarray}
where we used $|\delta|\leq \pi\varepsilon\tau/2$. Then, plugging this into Eq.~\eqref{eq:diff}, we have:
\begin{eqnarray}
\abs{\lambda_2(G(x))-\tilde\lambda} &\leq & \frac{3\varepsilon}{4}n\pi^2\tau^2\nonumber\\
&=& \frac{3\varepsilon}{4}n\pi^2\frac{\Delta(U(P,x))^2}{\pi^2} 
\;\leq \; \frac{3\varepsilon}{4}n\pi^2\sin^2\left(\frac{\Delta(U(P,x))}{2}\right)
\;=\; \pi^2\frac{3\varepsilon}{4}\lambda_2(G(x)),
\end{eqnarray}
using the fact that $\frac{x^2}{\pi^2}\leq\sin^2(x/2)$ when $x\in [-\pi,\pi]$. 
\end{proof}

\cref{thm:Connectivity-estimation} now follows, restated below for convenience.

\ConnEstimation*

\section{Acknowledgments} 
SJ is supported by an NWO WISE Grant and NWO Veni Innovational Research Grant under project number 639.021.752.
SK completed some of this work while at the Joint Center for Quantum Information and Computer Science (QuICS) at the
University of Maryland. This research was supported in part by Perimeter
Institute for Theoretical Physics. Research at Perimeter Institute is
supported by the Government of Canada through Industry Canada and by the
Province of Ontario through the Ministry of Economic Development and
Innovation.

\bibliographystyle{alpha}
\bibliography{capbib}

\appendix

\section{Effective Capacitance and Effective Conductance}\label{ap:effCapDef}

In this section, we provide more intuition for \cref{def:effCap} based on the
definition of effective capacitance and effective conductance. This section
roughly follows the explication of \cite[IX.1] {bollobas2013modern}.

Let $G$ be a graph with implicit weights $c$. As in \cref{sec:graphPrelim}, we
associate a subgraph $G(x)$ of $G$ with an electrical circuit in the following
way: we put a $0$-resistance wire at every edge $e\in E(G(x))$, and at every
edge $e\in E(G)\setminus E(G((x))$, we put a capacitor with capacitance
$c(e)$. Then we consider the effective capacitance of this circuit when a
voltage source is connected between $s$ and $t.$

When a voltage is applied between $s$ and $t$, some amount of charge flows
from the initially uncharged part of the circuit connected to $t$ to the
initially uncharged part of the circuit connected to $s$. Eventually, some
steady-state accumulation of charge on the capacitors is reached. Then the
effective capacitance is given by $Q/\sop E$, where $Q$ is the total amount of
charge that moves from the $t$ component to the $s$ component, and $\sop E$ is
the voltage applied to the battery.

We will show that $Q/\sop E$ is equal to the definition of effective
capacitance given in \cref{def:effCap}. Without loss of generality, we set
$\sop E=1.$ To determine $Q$, we first use the fact that charge is conserved.
We define a function $q:\overrightarrow{E}(G)\rightarrow
\mathbb{R}$, that tracks how charge moves in the circuit. For $(u,v,l)\in \overrightarrow{E}(G(x))$, we define 
$q(u,v,\edgeL)$ to be the amount of charge that is shifted through the edge from
$u$ to $v$ from the time that the battery is connected until a steady state
is reached. For  $(u,v,\edgeL)\in E(G)\setminus E(G(x))$, we define
$q(u,v,\edgeL)$ to be the amount of charge that accumulates on the
capacitor across $(\{u,v\},\edgeL)$, specifically on the side of the capacitor closest to vertex $u$.

Then by conservation of charge, we have
\begin{enumerate}
\item $\forall (u,v,\edgeL)\in \overrightarrow{E}(G)$, $q(u,v,\edgeL )=-q(v,u,\edgeL)$;
\item $\sum_{v,\edgeL:(s,v,\edgeL)\in \overrightarrow{E}(G)}q(s,v,\edgeL)=\sum_{v,\edgeL:(v,t,\edgeL)\in \overrightarrow{E}(G)}q(v,t,\edgeL)=Q$; and 
\item $\forall  u\in V(G)\setminus\{s,t\}$, $\sum_{v,\edgeL:(u,v,\edgeL )\in \overrightarrow{E}(G)}q(u,v,\edgeL )=0$. 
\end{enumerate}

We define $\overline{\sop V}:V(G)\rightarrow \mathbb{R}$ to be the voltage at
each vertex in the circuit at steady state, where without loss of generality,
we set $\overline{\sop V}(t)=0$. Then $\overline{\sop V}$ is a unit
$st$-potential. To see this, note first that $\overline{\sop V}(s)=1$ because
the voltage difference between $s$ and $t$ must be 1, and the voltage
difference across edges in $E(G(x))$ must be zero because these vertices are
connected by $0$-resistance wires.

However, $\overline{\sop V}$ must also satisfy the capacitance ratio across
each individual edge with a capacitor. That is, for each edge $(u,v,\edgeL)\in
\overrightarrow{E}(G)\setminus \overrightarrow{E}(G(x))$,
\begin{align}\label{eq:smallCap}
c(u,v,\edgeL)=\frac{q(u,v,\edgeL)}{\overline{\sop V}(u)-\overline{\sop V}(v)}.
\end{align}

Rearranging terms, applying the conservation of charge condition, and using the
fact that for $(u,v,\edgeL)\in \overrightarrow{E}(G(x))$, we have $\overline{\sop
V}(u)-\overline{\sop V}(v)=0$, we find that for each $u\in V(G)\setminus\{s,t\}$,
\begin{align}\label{eq:conserve}
\sum_{v,\edgeL:(u,v,\edgeL )\in \overrightarrow{E}(G)}\left(\overline{\sop V}(u)-\overline{\sop V}(v)\right)c(u,v,\edgeL)=0. 
\end{align}
Looking at \cref{eq:conserve}, we see that
\begin{align}\label{eq:CappPre}
\overline{\sop V}=\argmin_{\sop V}\frac{1}{2}\sum_{(u,v,\edgeL)\in \overrightarrow{E}(G)}\left(\sop V(u)-\sop V(v)\right)^2c(u,v,\edgeL),
\end{align} 
where the minimization is over unit $st$-potentials on $G(x)$.
To see this, note that the minimum occurs when the derivative with respect to $\sop V(v)$ is zero. 

Therefore
\begin{align}\label{eq:Capp2}
\min_\sop V\frac{1}{2}\sum_{(u,v,\edgeL)\in \overrightarrow{E}(G)}\left(\sop V(u)-\sop V(v)\right)^2c(u,v,\edgeL)&=\frac{1}{2}\sum_{(u,v,\edgeL)\in \overrightarrow{E}(G)}\left(\overline{\sop V}(u)-\overline{\sop V}(v)\right)^2c(u,v,\edgeL)
\nonumber\\
&=\frac{1}{2}\sum_{(u,v,\edgeL)\in \overrightarrow{E}(G)}\left(\overline{\sop V}(u)-\overline{\sop V}(v)\right)q(u,v,\edgeL).
\end{align}
Using conservation of charge, \cref{eq:Capp2} becomes
\begin{align}
\sop V(s)\sum_{v,\edgeL:(s,v,\edgeL)\in \overrightarrow{E}(G)}q(s,v,\edgeL)-\sop V(t)\sum_{v,\edgeL:(v,t,\edgeL)\in \overrightarrow{E}(G)}q(v,t,\edgeL)=Q.
\end{align}
Thus
\begin{align}
\min_\sop V\frac{1}{2}\sum_{(u,v,\edgeL)\in \overrightarrow{E}(G)}\left(\sop V(u)-\sop V(v)\right)^2c(u,v,\edgeL)=Q/\sop E,
\end{align}
so the two notions of capacitance (from electrical circuits and \cref{def:effCap}) coincide. (Recall we have set $\sop E=1$.)

\begin{figure}[h!]
\centering
\begin{tikzpicture}
\node at (0,0) {\begin{tikzpicture}
\draw[dashed, red] (1,0) -- (.25,1) -- (1,.5) -- (1.1,1.2);

\draw[dashed, green] (1.1,1.2) -- (2,.7);

\draw[dashed, blue] (1.2,1.7) -- (2.2,1.7) -- (1.1,1.2);

\draw[dashed, purple] (2.2,1.7) -- (2.9,1);

\draw[dashed, orange] (2,.7) -- (1.5,.4) -- (2.4,.1);

\filldraw (.4,1.6) circle (.1);
\filldraw (1.2,1.7) circle (.1);
\filldraw (1.1,1.2) circle (.1);
\filldraw (.25,1) circle (.1);
	\draw[thick] (.25,1) -- (.4,1.6) -- (1.2,1.7) -- (1.1,1.2) -- (.4,1.6);

\filldraw (2.2,1.7) circle (.1);

\filldraw (2,.7) circle (.1);
\filldraw (2.5,.7) circle (.1);
\filldraw (2.9,1) circle (.1);
\filldraw (2.4,.1) circle (.1);
	\draw[thick] (2,.7) -- (2.5,.7)--(2.9,1);
	\draw[thick] (2.4,.1) -- (2.5,.7);

\filldraw (1,.5) circle (.1); 	\filldraw (1.5,.4) circle (.1);
\filldraw (1,0) circle (.1);
	\draw[thick] (1,.5) -- (1.5,.4) -- (1,0) -- (1,.5);
	\node at (1.5,-.7) {$G(x)$};
\end{tikzpicture}};

\node at (5,0) {\begin{tikzpicture}
\draw[blue] plot[smooth] coordinates {(0,1) (.5,1.1) (1,1)};
\draw[blue] plot[smooth] coordinates {(0,1) (.5,.9) (1,1)};

\draw[red] plot[smooth] coordinates {(0,1) (-.2,.5) (0,0)};
\draw[red] (0,1) -- (0,0);
\draw[red] plot[smooth] coordinates {(0,1) (.2,.5) (0,0)};

\draw[green] (0,1) -- (1,0);

\draw[purple] (1,1) -- (1,0);

\draw[orange] plot[smooth] coordinates {(0,0) (.5,.1) (1,0)};
\draw[orange] plot[smooth] coordinates {(0,0) (.5,-.1) (1,0)};

\filldraw (0,1) circle (.1); \filldraw (1,1) circle (.1);
\filldraw (0,0) circle (.1); \filldraw (1,0) circle (.1);
	\node at (.5,-1.2) {$G(x)^\complement$};
\end{tikzpicture}};

\end{tikzpicture}
\caption{Edges in $G(x)$ are shown using solid lines, while edges in $G\setminus G(x)$ are shown using dashed lines. In $G(x)^\complement$, each of the four connected components of $G(x)$ becomes a vertex, and the number of edges between vertices depends on the number of edges in $G\setminus G(x)$ connecting one component to another.}\label{fig:capacitance-graph}
\end{figure}
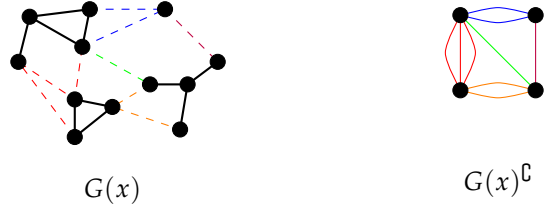

Other readers may recognize \cref{def:effCap} as the formula for effective
conductance of a network (see e.g. \cite{bollobas2013modern}), where the
effective conductance is the inverse of the effective resistance of a graph.
However, the network for which \cref{def:effCap} is the effective conductance
is a bit strange. Given a network $\sop N=(G,c)$ and a subgraph $G(x)$, create
a network $(G(x)^\complement, c^\complement)$ where each connected component
$w$ in $G(x)$ corresponds to a vertex $v_w$ in $G(x)^\complement$. Then for
every edge $(u,v,\edgeL)\in E(G)\setminus E(G(x))$ such that $u$ is in one
connected component $w$ in $G(x)$, and $\eta$ is in another connected
component $y$ in $G(x)$, create an edge in $G(x)^\complement$ between $v_w$
and $v_y$ with weight
$c^\complement(\{v_w,v_y\},\edgeL)=c(\{u,\eta\},\edgeL).$ Let $v_s$ be the
connected component containing $s$ and $v_t$ be the connected component
containing $t$. Then the effective conductance of $G(x)^\complement$ between
$v_s$ and $v_t$ is given by \cref{def:effCap}. \cref{fig:capacitance-graph}
shows the correspondence between a graph $G(x)$ and $G(x)^\complement$.


\end{document}